\documentclass[11pt]{article}
\usepackage{latexsym,amssymb,amsmath,amsthm,amscd}
\usepackage{algorithmic}
\usepackage{amssymb}
\usepackage{times}
\usepackage{enumerate}
\usepackage[T1]{fontenc}
\usepackage[usenames,dvipsnames]{pstricks}
\usepackage{epsfig}
\usepackage{url}
\usepackage{hyperref}

\def\QED{\hfill$\Box$}

\renewenvironment{proof}{\noindent {\bfseries Proof. }}{\QED}

\newtheorem{theorem}{Theorem}[section]
\newtheorem{lemma}       [theorem]{Lemma}
\newtheorem{corollary}   [theorem]{Corollary}
\newtheorem{proposition} [theorem]{Proposition}
\newtheorem{remark}      [theorem]{Remark}
\newtheorem{example}     [theorem]{Example}

\newtheorem{definition}  [theorem]{Definition}
\newtheorem{question}    [theorem]{Question}
\newtheorem{algorithm}   [theorem]{Algorithm}
 \newtheorem{th-algorithm}[theorem]{Theorem}

\newtheorem{model}       [theorem]{Model}

\DeclareMathOperator{\Mult}{M}
\newcommand{\mult}[2][a_1]{\ensuremath{\Mult_{#1} \left( #2 \right)}}

\newcommand{\Wr}{\ensuremath \text{Wr}}
\newcommand{\Sparsest}{\ensuremath \text{Sparsest}}
\newcommand{\Waring}{\ensuremath \text{Waring}}
\newcommand{\AffPow}{\ensuremath \text{AffPow}}
\newcommand{\SDE}{\ensuremath \text{SDE}}

\newcommand{\KK}{\mathbb{K}}
\newcommand{\FF}{\mathbb{F}}
\newcommand{\TT}{\mathbb{T}}
\newcommand{\NN}{\mathbb{N}}
\newcommand{\RR}{\mathbb{R}}

\newcommand{\CC}{\mathbb{C}}
\newcommand{\ZZ}{\mathbb{Z}}

\title{Reconstruction Algorithms\\ for Sums of Affine Powers}

\author{Ignacio Garc\'ia-Marco, Pascal Koiran, Timoth\'ee Pecatte\\
LIP\thanks{UMR 5668 Ecole Normale Sup\'erieure de  Lyon, CNRS, UCBL, INRIA.
The authors are supported by ANR project
CompA (code ANR--13--BS02--0001--01).
Email:  
{\tt [Pascal.Koiran, Timothee.Pecatte]@ens-lyon.fr,} 
{\tt iggarcia@ull.es}.
}, 
Ecole Normale Sup\'erieure de Lyon, Universit\'e de Lyon.
}

\begin{document}
\maketitle

\begin{abstract}
A sum of affine powers is an expression of the form		
$$f(x) = \sum_{i=1}^s \alpha_i (x - a_i)^{e_i}.$$
Although quite simple, this model is a generalization of 
two well-studied models: Waring decomposition and Sparsest Shift.
For these three models there are natural extensions to several variables, 
but this paper is mostly focused on univariate polynomials.
We present structural results which compare the expressive power of 
the three models; 
and we  propose algorithms that find the smallest decomposition
of $f$ in the first model (sums of affine powers) for an input polynomial $f$
given in dense representation.
We also begin a study of the multivariate case.

This work could be extended in several directions.
In particular, just as for Sparsest Shift and Waring decomposition,
one could consider extensions to ``supersparse'' polynomials and attempt 
a fuller
study of the multivariate case.
We also point out that the basic univariate problem studied in 
the present paper is far from completely solved: our algorithms all rely on 
some assumptions for the exponents $e_i$ in a decomposition of $f$, and some
algorithms also rely on a distinctness assumption for the shifts $a_i$.
It would be very interesting to weaken these assumptions, or even to remove
them entirely. Another related and poorly understood issue is that of
the bit size of the constants $a_i,\alpha_i$ in an optimal decomposition:
is it always polynomially related to the bit size of the input polynomial 
$f$ given in dense representation?
\end{abstract}

\section{Introduction}

Let $\FF$ be any characteristic zero field and let $f \in \FF[x]$ be a univariate polynomial.
This work concerns the study of expressions of $f$ as a linear combination of powers of affine forms.

\begin{model} \label{affpowmodel}
	We consider expressions of $f$ of the form:
	\[
		f = \sum_{i=1}^s \alpha_i (x - a_i)^{e_i}
	\]
	with $\alpha_i, a_i \in \FF$, $e_i \in \NN$. We denote by $\AffPow_\FF(f)$ the minimum value $s$ such that there exists a representation of the previous form with $s$ terms.
\end{model}

This model was already studied in \cite{GK}, where we gave explicit examples of
polynomials of degree $d$ requiring $\AffPow_\RR(f)=\Omega(d)$ terms for the 
field $\FF=\RR$.

The main goal of this work is to design algorithms that reconstruct 
the optimal representation of polynomials in this model, i.e., algorithms
that receive as input $f \in \FF[x]$ and compute the exact value $s = \AffPow_\FF(f)$ and 
a set of triplets of coefficients, nodes and exponents $\{(\alpha_i,a_i,e_i)\, \vert \, 1 \leq i \leq s\} \subseteq \FF \times \FF \times \NN$
such that $f = \sum_{i = 1}^s \alpha_i (x-a_i)^{e_i}$.
We assume that $f$ is given in dense representation, i.e., as a tuple of 
$\deg(f)+1$ elements of $\FF$.

\medskip

Model~\ref{affpowmodel} extends two already well-studied models. 
The first one is the Waring model, where all the exponents are equal to the degree of the polynomial, i.e.,  $e_i = \deg(f)$ for all $i$.
\begin{model} \label{waringmodel}
	For a polynomial $f$ of degree $d$, we consider expressions of $f$ of the form:
	\[
		f =  \sum_{i=1}^s \alpha_i (x - a_i)^d
	\]
	with $\alpha_i, a_i \in \FF$.
	We denote by $\Waring_\FF(f)$ the \emph{Waring rank} of $f$, which is the minimum value $s$ such that there exists a representation of the previous form with $s$ terms.
\end{model}
Waring rank has been studied by algebraists and geometers since the 19th century. The algorithmic study of Model~\ref{waringmodel} is usually attributed to Sylvester. We refer to~\cite{IaKa} for the historical background and to section~1.3 of that book for a description of the algorithm (see also Kleppe~\cite{Kleppe} and Proposition~46 of Kayal~\cite{Kayal12}). Most of the subsequent work was devoted to the multivariate generalization\footnote{In the literature, Waring rank is usually defined for homogeneous polynomials. After homogenization, the univariate model~\ref{waringmodel} becomes bivariate and the ``multivariate generalization'' 
therefore deals with homogeneous polynomials in 3 variables or more.} of Model~\ref{waringmodel}, with much of the 20th century work focused on the determination of the Waring rank of generic polynomials~\cite{alexander95,brambilla08,IaKa}. 
A few  recent papers~\cite{landsberg2010,BCG} 
have begun to investigate the Waring rank of specific 
polynomials such as monomials, sums of coprime monomials, 
the permanent and the determinant.

The second model that we generalize is the Sparsest Shift model, where all the shifts $a_i$ are required to be equal.
\begin{model}
	For a polynomial $f$, we consider expressions of $f$ of the form:
	\[
		f =  \sum_{i=1}^s \alpha_i (x - a)^{e_i}
	\]
	with $\alpha_i, a \in \FF, e_i \in \NN$.
	We denote by $\Sparsest_\FF(f)$ the minimum value $s$ such that there exists a representation of the previous form with $s$ terms.
\end{model}
This model and its variations have been studied in the computer 
science literature 
at least since Borodin and Tiwari~\cite{BoTi}.
Some of these papers deal with multivariate generalizations~\cite{GK93,GKL}, 
with ``supersparse'' polynomials\footnote{In that model, the size of the monomial $x^d$ is defined to be $\log d$ instead of $d$ as in the usual dense encoding.}~\cite{GR} 
 or establish condition for the uniqueness of the sparsest shift~\cite{LS}.
It is suggested at the end of~\cite{GKL} to allow ``multiple shifts'' instead
of a single shift, and this is just what we do in this paper. More precisely,
as is apparent from Model~\ref{affpowmodel}, we do not place any constraint on the number of distinct shifts: it can be as high as the number $s$ of affine powers.
It would also make sense to place an upper bound $k$ on the number 
of distinct shifts. This would provide a smooth interpolation between 
the sparsest shift model (where $k=1$) and Model~\ref{affpowmodel}, where $k=s$.

\subsection{Our results}

We provide both structural and algorithmic results. Our structural results are
presented in Section~\ref{structure}. We compare the expressive power of our 3 models: sums of affine powers, sparsest shift and the Waring decomposition.
Namely, we show that some polynomials have a much smaller
expression as a sum of affine powers than in the sparsest shift or Waring 
models. Moreover, we show that the Waring and sparsest shift models are ``orthogonal'' in the sense that (except in one trivial case) no polynomial can 
have a small representation in both models at the same time.
We also show that some real polynomials have a short expression 
as a sum of affine powers over the field of complex numbers, but not over 
the field of real numbers.
 Finally, we study the uniqueness of 
the optimal representation  as a sum of affine powers.
It turns out that our reconstruction algorithms 
 only work in a regime where the uniqueness of optimal
representations is guaranteed.

As already explained, we present algorithms that find 
the optimal representation of an input polynomial $f$. 
We achieve this goal in several cases, but we do not solve the problem 
in its full generality. One typical result is as follows (see Theorem~\ref{diffnodesimproved} in Section~\ref{distinctsec} for a more detailed statement which includes a description of the algorithm).
\begin{theorem} \label{diffnodesintro}
	Let $f \in \FF[x]$ be a polynomial that can be written as 
	\[
		f = \sum_{i = 1}^s \alpha_i (x - a_i)^{e_i},
	\]
	where the constants $a_i \in \FF$ are all distinct, $\alpha_i \in \FF \setminus \{0\}$, and $e_i \in \NN$. Assume moreover that $n_i \leq   (3i/4)^{1/3}
	  - 1$ for all 
$i \geq 2$, where $n_i$ denotes the number of indices $j$ such that
 $e_j \leq i$.

	Then, $\AffPow_\FF(f) = s$. Moreover, there is a polynomial time algorithm that receives $f = \sum_{i = 0}^d f_i x^i \in \FF[x]$ as input and computes the $s$-tuples of coefficients $C(f) = (\alpha_1,\ldots,\alpha_s)$, of nodes $N(f) = (a_1,\ldots,a_s)$ and exponents $E(f) = (e_1,\ldots,e_s)$.
\end{theorem}
From the point of view of the optimality of representations, 
it is quite natural
to assume an upper bound on the numbers $n_i$. 
Indeed, if there is an index $j$ such that $n_j > j+1$ 
then the powers $(x-a_i)^{e_i}$ are linearly dependent, and there would be a 
smaller expression of $f$ as a linear combination of these polynomials.\footnote{It is hardly more difficult to show that one must have 
$n_j \leq \lceil \frac{j+1}{2} \rceil$ for any optimal expression, see~\cite[Proposition 18]{GK}.}
We would therefore  have  $\AffPow_\FF(f) < s$ instead of 
$\AffPow_\FF(f) =s$. It would nonetheless be interesting to relax the
assumption $n_i \leq (3i/4)^{1/3} - 1$ in this theorem. 
Another restriction is the assumption that the shifts $a_i$ are all distinct.
We relax that assumption in Section~\ref{repeatsec} 
{  but we still need to assume that all the exponents $e_i$ corresponding to a 
given shift $a_i=a$ belong to a ``small'' interval 
(see Theorem~\ref{th-algVeryLargeExp} for a precise statement).}
Alternatively, we can assume instead that there is a large gap 
between the exponents in two consecutive occurences of the same shift as in Theorem~\ref{th-algbiggaps}.

In Section~\ref{multisec} we extend the sum of affine powers model to several variables. We consider expressions of the form
\begin{equation} \label{multimodel}
f(x_1,\ldots,x_n)= \sum_{i = 1}^s \alpha_i \ell_i(x_1,\ldots,x_n)^{\,e_i},
\end{equation}
where $e_i \in \NN$, $\alpha_i \in \FF$ and $\ell_i$ is a (non constant) linear form for all $i$. This is clearly a generalization of the univariate model~\ref{affpowmodel} and of multivariate Waring decomposition.
Work on multivariate sparsest shift has developed in a different direction:
one idea~\cite{GKL} has been to transform the input polynomial into a sparse polynomial 
by applying a (possibly) different shift to each variable. The model from~\cite{GK93} is more general than~\cite{GKL}, and we do {\em not} generalize any of these two models.
Our algorithmic strategy for reconstructing expressions 
of the form~(\ref{multimodel}) is to transform the multivariate problem
into univariate problems by projection, and to ``lift'' the solution
of $n$ different projections to the solution of the multivariate problem.
This can be viewed as an analogue of ``case 1'' of Kayal's algorithm for
Waring decomposition~\cite[Theorem 5]{Kayal12}.

\subsection{Main tools} \label{tools}

Most of our results\footnote{The structural results about real polynomials from Section~\ref{realsec} rely instead on Birkhoff interpolation~\cite{GK}.} hinge on the study of certain differential equations
satisfied by the input polynomial $f$. We consider differential equations 
of the form 
\begin{equation} \label{SDE}
\sum_{i=0}^k P_i(x)f^{(i)}=0
\end{equation}
 where the $P_i$'s are polynomials.
If the degree of $P_i$ is bounded by $i+l$ for every $i$, 
we say that~(\ref{SDE}) is a  {\em Shifted Differential Equation (SDE)} of
order $k$ and shift $l$. 
Section~\ref{prelim} recalls
some (mostly standard) background on differential equations and the 
Wronskian determinant.

 When $f$ is a polynomial with an expression of size~$s$
  in Model~\ref{affpowmodel} we prove in Proposition~\ref{smallequation}
  that $f$ satisfies a ``small'' SDE,
of order $2s-1$ and shift zero.
The basic idea behind our algorithms is to look for one of these ``small'' SDEs satisfied
by $f$, and hope that the powers $(x-a_i)^{e_i}$ in an optimal decomposition
of $f$ satisfy
the same SDE. This isn't just wishful thinking because the SDE 
from  Proposition~\ref{smallequation} is satisfied not only by $f$ but also
by the powers $(x-a_i)^{e_i}$.

Unfortunately, this basic idea  by itself does not yield efficient 
algorithms. The main difficulty is that $f$ could satisfy several SDE of
order $2s-1$ and shift 0. By Remark~\ref{findSDE} 
we can efficiently find such a SDE, 
but what if we don't find the ``right'' SDE, 
i.e., the SDE which (by Proposition~\ref{smallequation}) is guaranteed to be satisfied by $f$ {\em and} by the powers $(x-a_i)^{e_i}$?
One way around this difficulty is to assume that the exponents $e_i$ are all
sufficiently large compared to $s$. In this case we can show that every SDE of order $2s-1$ and shift $0$ which is satisfied by $f$
 is also satisfied by $(x-a_i)^{e_i}$.
This fact is established in Corollary~\ref{bigexpCor}, and yields the 
following result (see Theorem~\ref{th-algbigexp} in Section~\ref{distinctsec} for a more detailed statement which includes a description of the algorithm).
\begin{th-algorithm}[Big exponents]\label{bigexpintro}
	Let $f \in \FF[x]$ be a polynomial that can be written as 
	\[
		f = \sum_{i = 1}^s \alpha_i (x - a_i)^{e_i},
	\]
 	where the constants $a_i \in \FF$ are all distinct, $\alpha_i \in \FF \setminus \{0\}$ and $e_i > 5 s^2/2$. 
 	Then, $\AffPow_\FF(f) = s$. Moreover, there is a polynomial time algorithm that receives $f = \sum_{i = 0}^d f_i x^i \in \FF[x]$ as input and computes 
 	the $s$-tuples of coefficients  $C(f) = (\alpha_1,\ldots,\alpha_s)$, of nodes $N(f) = (a_1,\ldots,a_s)$ and exponents $E(f) = (e_1,\ldots,e_s)$. 
\end{th-algorithm}
The algorithm of Theorem~\ref{diffnodesintro} is more involved: 
contrary to Theorem~\ref{bigexpintro}, we cannot 
determine all the terms $(x-a_i)^{e_i}$ in a single pass. Solving the SDE
only allows the determination of some (high degree) terms. We must then 
subtract these terms from $f$, and iterate.

In the first version of this paper,\footnote{\href{http://arxiv.org/abs/1607.05420v1}{arxiv.org/abs/1607.05420v1}}
instead of a $\SDE(2s-1,0)$ we used a SDE of order $s$ and shift $s \choose 2$ originating from the Wronskian determinant
(compare the two versions of Proposition~\ref{smallequation}).
Switching to the new SDE led to significant improvements in most of our algorithmic results. 
For instance, in the first version of Theorem~\ref{bigexpintro} the exponents
$e_i$ had to satisfy the condition $e_i > s^2(s+1)/2$ instead of the current (less stringent) condition
$e_i > 5s^2/2$.

\subsection{Models of computation} \label{modcomp}

 Our algorithms take as inputs polynomials with coefficients in an arbitrary field $\KK$ of characteristic 0. At this level of generality, we need to be able to perform arithmetic operations (additions, multiplications) and equality tests
between elements of $\KK$. 
When we write that an algorithm runs in polynomial time, we mean that the number
of such steps is polynomial in the input size.
This is a fairly standard setup for algebraic algorithms (it {  is}
also interesting to analyze the bit complexity of our algorithms for some specific fields such as the field of rational numbers; 
more on this {  at the end of this subsection and} 
in Section~\ref{future}). 
An input polynomial of degree $d$ is represented simply by the list of coefficients of its $d+1$ monomials, and its size thus equals $d+1$.
In addtion to arithmetic operations and equality tests, we need to to be able to compute roots of polynomials with coefficients in $\KK$. This is in general unavoidable: for an an optimal decomposition  of $f \in \KK[X]$ in Model~\ref{affpowmodel}, the coefficients $\alpha_i, a_i$ 
may lie in an extension field $\FF$ of $\KK$
(see Section~\ref{structure} and more precisely Example~\ref{realvscomplex} in Section~\ref{realsec} for the case $\KK=\RR, \FF=\CC$). If the optimal decomposition has size~$s$, we need to compute roots of polynomials of degree at most $2s-1$.\footnote{Except in the algorithm of Theorem~\ref{th-algVeryLargeExp}, where we need to compute roots of polynomials at most  $2s-1+\delta$. Here $\delta$ is a parameter of the algorithm, see Theorem~\ref{th-algVeryLargeExp} for details.}
As a rule, root finding is used only to output
the nodes $a_i$ of the optimal decomposition,\footnote{Once the $a_i$'s have been determined, we also need to do some linear algebra computations with these nodes to determine the coefficients $\alpha_i$.}
but the ``internal working'' of our algorithms 
remains purely rational 
(i.e., requires only arithmetic operations and comparisons).
This is similar to the symbolic algorithm for univariate sparsest shifts of
Giesbrecht, Kaltofen and Lee~(\cite{GKL}, p.~408 of the journal paper),
which also needs access to a polynomial root finder.

The one exception to this rule is the algorithm of 
Theorem~\ref{diffnodesintro}.
As mentioned at the end of Section~\ref{tools}, this is an iterative algorithm.
At each step of the iteration we have to compute roots of polynomials (which may lie outside $\KK$), and we keep computing with these roots in the subsequent iterations. For more details see Theorem~\ref{diffnodesimproved} and the discussion after that theorem.
We make a first step toward removing root finding from the internal working of this algorithm in Proposition~\ref{Galois}.

{  We also take some steps toward the analysis of our algorithms 
in the bit model of computation. 
We focus on the algorithm of Theorem~\ref{diffnodesintro} 
since it is the most difficult to analyze due to its iterative nature.
We show in Proposition~\ref{polyoutput} that for polynomials with integer coefficients, this algorithm can be implemented in the bit model
to run in time polynomial in the bit size of the {\em output}. We do not have a polynomial
running time bound as a function of the input size (more on this in Section~\ref{future}).}
We also compute  explicitly a polynomial
  bound on the running time of  the simpler algorithm of Theorem~\ref{th-algbigexp}, which deals with the case of ``big exponents''. Our bound is of fairly
  large degree and is probably not optimal.

\subsection{Future work} \label{future}

One could try to extend the results of this paper in several directions.
For instance, one could try to handle ``supersparse'' polynomials like in the
Sparsest Shift algorithm of~\cite{GR}.
The multivariate case would also deserve further study. As explained above
we proceed by reduction to the univariate case, but one could try to design
more ``genuinely multivariate'' algorithms. For Waring decomposition, 
such an algorithm is proposed in ``case 2'' of~\cite[Theorem 5]{Kayal12}.
Its analysis relies on a randomness assumption for the input $f$ (our multivariate algorithm is randomized, but in this paper we never assume that the 
input polynomial is randomized).

One should also keep in mind, however, that the basic univariate 
problem studied in 
the present paper is far from completely solved: our algorithms all rely on 
some assumptions for the exponents $e_i$ in a decomposition of $f$, and some
algorithms also rely on a distinctness assumption for the shifts $a_i$.
It would be very interesting to weaken these assumptions, or even to remove
them entirely. With a view toward this question, one could first try to 
improve the lower bounds from~\cite{KKPS}. Indeed, the same tools 
(Wronskians, shifted differential equations) turn out to be useful for the 
two problems (lower bounds and reconstruction algorithms) but the lower
bound problem appears to be easier. For real polynomials we have already 
obtained optimal $\Omega(d)$ lower bounds in~\cite{GK} 
using Birkhoff interpolation, but it remains to give an algorithmic application
of this lower bound method.

Another issue that we have {  only begun to address}
is the analysis of the
bit complexity of our algorithms. 
We give an explicit polynomial bound on the bit complexity of
the algorithm of Theorem~\ref{th-algbigexp},
but this issue seems
to be more subtle for Theorem~\ref{diffnodesintro} 
due to the iterative nature of our algorithm. It is in fact not clear
that there exists a solution of size polynomially bounded in the input size
(i.e., in the bit size of $f$ given as a sum of monomials).
More precisely, we ask the following question.
\begin{question}\label{questinputoutput}
We define the dense size of a polynomial $f = \sum_{i=0}^d f_ix^i \in \ZZ[X]$
as $\sum_{i=0}^d [1+\log_2(1+|f_i|)]$.
Assume that $f$ can be written as 	
$$f = \sum_{i = 1}^s \alpha_i (x - a_i)^{e_i}$$ 
with $a_i \in \ZZ$, $\alpha_i \in \ZZ \setminus \{0\}$, 
and that this decomposition satisfies the conditions of Theorem~\ref{diffnodesintro}:  the constants $a_i$ are all distinct, and $n_i \leq  (3i/4)^{1/3} - 1$ for all 
$i \geq 2$, where $n_i$ denotes the number of indices $j$ such that
 $e_j \leq i$.

Is it possible to bound the bit size of the constants $\alpha_i, a_i$ by
a polynomial function of the dense size of $f$ ?
\end{question}
{  As explained at the end of Section~\ref{modcomp}, 
under the same conditions we have a decomposition algorithm that runs 
in time polynomial in the bit size of the {\em output}. It follows that the above question has a positive answer if and only if 
there is a decomposition algorithm that runs in time polynomial 
in the bit size of the input 
(i.e., in time polynomial in the dense size of $f$).}

One could also ask similar questions in the case where 
the conditions of Theorem~\ref{diffnodesintro} do not hold. For instance, 
assuming that $f$ has an optimal decomposition with integer coefficients, is there such a decomposition where the coefficients $\alpha_i, a_i$ are of size polynomial in the size of $f$ ?

\section{Preliminaries} \label{prelim}

In this section we present some tools that are useful for their algorithmic applications in Sections~\ref{distinctsec} and~\ref{repeatsec}.
Section~\ref{structure} can be read independently, except for the proof of 
Proposition~\ref{LowerBound} and Theorem \ref{th-algVeryLargeExp} which use the Wronskian.

\subsection{The Wronskian}

In mathematics the \emph{Wronskian} is a tool mainly used in the study of differential equations, where it can be used to show that a set of solutions is linearly independent.
\begin{definition}[Wronskian]
	For $n$ univariate functions $f_1, \ldots, f_n$, which are $n-1$ times differentiable, the Wronskian $\Wr(f_1, \ldots, f_n)$ is defined as
	\[
		\Wr(f_1, \ldots, f_n)(x) =
		\begin{vmatrix}
		f_1(x) 		& 	f_2(x) 		& \dots  & f_n(x)		\cr 
		f_1'(x) 	& 	f_2'(x) 	& \dots  & f_n'(x)		\cr 
			\vdots 	& 	\vdots 		& \ddots & \vdots 		\cr
		f_1^{(n-1)}	& f_2^{(n-1)}	& \dots	 & f_n^{(n-1)}
		\end{vmatrix} 
	\]
\end{definition}

It is a classical result, going back at least to~\cite{Bocher}, that the Wronskian captures the linear dependence of polynomials in $\FF[x]$.
\begin{proposition} \label{indepWronskian}
	For $f_1, \ldots, f_n \in \FF[X]$, the polynomials are linearly dependent if and only if the Wronskian $\Wr(f_1, \ldots, f_n)$ vanishes everywhere.
\end{proposition}

For every $f \in \FF[x]$ and every $a \in \FF$ we denote by $\mult[a]{f}$ the multiplicity 
of $a$ as a root of $f$,
i.e., $\mult[a]{f}$
is the maximum $t \in \NN$ such that $(x-a)^t$ divides $f$.
The following
result from \cite{WronskianLemma} gives a Wronskian-based bound on the multiplicity of a root in a sum of polynomials.

\begin{lemma}\label{lemma:WrMult}
	Let $f_1, \ldots, f_n$ be some linearly independent polynomials and $a \in \FF$, and let $f(x) = \sum_{j=1}^n f_j(x)$.
	Then:
	\[
		\mult[a]{f} \leq n - 1 + \mult[a]{\Wr(f_1, \ldots, f_n)},
	\]
	where $\mult[a]{f}$ is finite since $\Wr(f_1,\ldots,f_n) \not\equiv 0$.
\end{lemma}

In \cite{PS76} one can find several properties concerning the Wronskian (and which have been known since the 19$^{th}$ century). In this work we will use
the following properties, which can be easily derived 
from those of \cite{PS76}. 
For the sake of completeness we include a short proof. 

\medskip

\begin{proposition}\label{factorWronskian}
	Let $f_1, \ldots, f_n \in \FF[x]$ be linearly independent polynomials and let $a \in \FF$. If $f_j = Q_j^{d_j} g_j$ with $Q_j \in \FF[x]$ and $d_j \geq n$ for
	all $j$, then $Q := \prod_{j = 1}^n Q_j^{\,d_j - n + 1}$ divides $\Wr(f_1,\ldots,f_n)$.
	Moreover, if $Q(a) \neq 0$, then 
		\[ 
		\mult[a]{\Wr(f_1,\ldots,f_n)} \leq \sum_{j = 1}^n \big[ {\rm deg}(g_j) + (n - 1) {\rm deg}(Q_j) \big] - \binom{n}{2}.
		\]
	Hence, if we set $f := \sum_{j = 1}^n f_j$, then 
		\[
		\mult[a]{f} \leq n - 1 +  \sum_{j = 1}^n \big[ {\rm deg}(g_j) + (n - 1) {\rm deg}(Q_j) \big] - \binom{n}{2}.
		\]		
\end{proposition}
\begin{proof}
Consider the $n \times n$ Wronskian matrix $W$ whose $(i+1,j)$-th entry is $f_j^{(i)}(x)$ with $0 \leq i \leq n-1$, $1 \leq j \leq n$. Since $Q_j^{d_j}$ divides $f_j$, then  $f_j^{(i)} 
	= Q_j^{d_j - i} g_{i,j} = Q_j^{d_j - n + 1} Q_j^{n-1-i} g_{i,j}$, for some $g_{i,j} \in \FF[x]$ of degree ${\rm deg}(g_j) + i\,{\rm deg}(Q_j) - i$.
	Since $Q_j^{d_j - n + 1}$ divides every element in the $j$-th
	column of $W$, we can factor it out from the Wronskian. 
This proves that $Q$ divides $\Wr(f_1,\ldots,f_n)$. Once we have factored out $Q_j^{d_j-n+1}$ for all $j$, we observe that $
	\Wr(f_1,\ldots,f_n) = Q(x)\, h(x)$, where
	$h(x)$ is the determinant of a matrix whose $(i+1,j)$-th entry has degree ${\rm deg}(g_j) + (n-1){\rm deg}(Q_j) - i$ for all $0 \leq i \leq n-1$ and $1 \leq j \leq n$.
	Hence, ${\rm deg}(h) \leq  \sum_{j = 1}^n \left[{\rm deg}(g_j) + (n-1){\rm deg}(Q_j)\right] - \binom{n}{2}.$
	Finally, we observe that if $Q(a) \neq 0$:
		$$\mult[a]{\Wr(f_1,\ldots,f_n)}= \mult[a]{Q} + \mult[a]{h} = \mult[a]{h} \leq {\rm deg}(h).$$
	For $f = \sum_{j = 1}^n f_j$, the upper bound for $\mult[a]{f}$ follows directly from Lemma \ref{lemma:WrMult}.
\end{proof}

\medskip
We observe that the result holds when some of the $Q_i(x) = 1$.

\subsection{Shifted Differential Equations}
\begin{definition}
A	\emph{Shifted Differential Equation} (SDE) is 
a differential equation of the form
	\[
		\sum\limits_{i=0}^k P_i(x)f^{(i)}(x) = 0
	\]
where $f$ is the unknown function and the $P_i$ are polynomials in $\FF[x]$ 
with $\deg(P_i) \leq i + l$.\\
	The quantity $k$ is called the \emph{order} of the equation,
and the quantity $l$ is called the \emph{shift}.
We will usually denote such a differential equation by $\SDE(k,l)$.
\end{definition}

One of the key ingredients for our results is that if $\AffPow(f)$ is small, then $f$ satisfies a ``small'' SDE. More precisely:

\begin{proposition}\label{smallequation}
	Let $\delta \in \ZZ^+$ and let $f \in \FF[x]$ be written as
	\[
		f = \sum_{i = 1}^t Q_i(x) (x - a_i)^{e_i}.
	\]
	where $a_i \in \FF$, $e_i \in \NN$ and ${\rm deg}(Q_i(x)) \leq \delta$ for all $i$.

        Then, $f$ satisfies a $\SDE(2t-1,\delta)$ which is also satisfied by the $t$ terms $f_i(x)=Q_i(x) (x - a_i)^{e_i}$.
        In particular, if $\AffPow_\FF(f) = s$, then
	 $f$ satisfies a $\SDE(2s-1,0)$.
\end{proposition}
\begin{proof}
  If we can find a $\SDE(2t-1,\delta)$ which is satisfied by all the $f_i$, by
  linearity the same $\SDE$ will be satisfied by $f$ and the theorem will be proved.
  The existence of this common $\SDE$ is equivalent to the existence of a nonzero
  solution for the following linear system in the unknowns $a_{j,k}$:
  $$\sum_{j,k} a_{j,k}x^jf_i^{(k)}(x)=0,$$
  where $1 \leq i \leq t$, $0 \leq k \leq 2t-1$ and $0 \leq j \leq k+\delta$.
  There are $(\delta + 1) + (\delta + 2) + \cdots + (\delta + 2t) = (2\delta + 2t+1)t$ unknowns, so we need to show that the matrix of this linear system
  has rank smaller than $(2\delta + 2t+1)t$.
  It suffices to show that for each fixed value of $i \in \{1,\ldots,t\}$,
  the subsystem:
  $$\sum_{j,k} a_{j,k}x^jf_i^{(k)}(x)=0\
  (0 \leq k \leq 2t-1, 0 \leq j \leq k+\delta)$$
  has a matrix of rank $<2\delta + 2t+1$.
  In other words, we have to show that the subspace $V_i$ spanned by
  the polynomials $x^jf_i^{(k)}(x)$ has dimension less than $2\delta + 2t+1$.
  But $V_i$ is included in the subspace spanned by the polynomials
  $$\{ (x-a_i)^{e_i  + j};\ -(2t-1) \leq j \leq \delta, e_i+j \geq 0\}.$$
  This is due to the fact that the polynomials $x^j$ and $Q_i(x)$
  belong respectively to the  spans of the polynomials
  $\{ (x-a_i)^{\ell}\, \vert \, 0 \leq \ell \leq j\},$ and
  $\{ (x-a_i)^{\ell}\, \vert \, 0 \leq \ell \leq \delta\}.$
  We conclude that $\dim V_i \leq 2t + \delta < 2\delta + 2t+1$.
\end{proof}

\begin{remark} \label{findSDE}
A polynomial $f$ satisfies a $\SDE(k,l)$ if and only if the polynomials
$(x^j f^{(i)}(x))_{0 \leq i \leq k, 0 \leq j \leq i+l}$ are linearly dependent
over $\FF$. The existence of such a $\SDE$ can therefore be decided efficiently 
by linear algebra, and when a $\SDE(k,l)$ exists it can be found explicitly by solving the corresponding linear system
{  (see, e.g., \cite[Corollary 3.3a]{Schrijver} for an analysis of linear system solving in the bit model of computation)}. We use this fact repeatedly
in the algorithms of Sections~\ref{distinctsec} and~\ref{repeatsec}.
\end{remark}

\medskip 
In this paper we will use some results concerning the set of solutions of a $\SDE$. They are particular cases of properties
that apply to linear homogeneous differential equations.

\medskip

\begin{lemma}\label{dimsolution}
	The set of polynomial solutions of a $\SDE$ of order $k$ is a vector space of dimension at most $k$.
\end{lemma}

\medskip
Given two $\SDE$ of order $k$: 
	$$\sum_{i=0}^k p_i(x)g^{(i)}(x) = 0 	\quad \text{ and } \quad  \sum_{i=0}^k q_i(x)g^{(i)}(x) = 0,$$
we say that they are equivalent if $p_k q_i = q_k p_i$ for all $i \in \{0,\ldots,k-1\}$. 
The following result can be found in \cite[Property 61]{PS76} and will only be used in the appendix.
We include a short proof.

\begin{lemma}\label{uniqueSDE}
For any set of $\FF$-linearly independent polynomials
	$f_1, \dots, f_k \in \FF[x]$, 
there exists a unique $\SDE$ (up to equivalence) of order $k$ 
satisfied simultaneously by all the $f_i$'s.
\end{lemma}
\begin{proof}
	Suppose there exist two different $\SDE$ satisfied by $f_1,\ldots,f_k$, namely: 
		$$\sum_{i=0}^k p_i(x)g^{(i)}(x) = 0	\quad \text{ and } \quad  \sum_{i=0}^k q_i(x)g^{(i)}(x) = 0.$$
	Then, we set $r_i := p_k q_i - q_k p_i$ for all $i \in \{0,\ldots,k\}$. By definition we have that $r_k = 0$ and we aim at proving that $r_i = 0$ for all $i$.
	Assume that there exists $j \in \{0,\ldots,k-1\}$ such that $r_j \neq 0$. Then, the following $\SDE$
		$$\sum_{i=0}^{k-1} r_i(x)\, g^{(i)}(x) = 0$$
	has order $\leq k-1$ and is satisfied by $f_1,\ldots,f_k$, a contradiction to Lemma \ref{dimsolution}.
\end{proof}

\section{Structural results} \label{structure}

In this section we compare the expressive power of our 3 models: sums of affine powers, sparsest shift and the Waring decomposition.
We will see in Section~\ref{charzero} that some polynomials have a much smaller
expression as a sum of affine powers than in the sparsest shift or Waring 
models. Moreover, we show that the Waring and sparsest shift models are ``orthogonal'' in the sense that (except in one trivial case) no polynomial can 
have a small representation in both models at the same time.

We begin this investigation of structural properties with the field of real 
numbers, where an especially strong version of orthogonality holds true.
We also show that some real polynomials have a short expression 
as a sum of affine powers over the field of complex numbers, but not over 
the field of real numbers. This observation has algorithmic implications: 
given a polynomial $f \in \FF[X]$, we may have to 
work in a field extension of $\FF$ to find the optimal representation for~$f$.
These ``real'' results can be derived fairly quickly from results in our 
previous paper~\cite{GK}.
We then move to arbitrary fields of characteristic zero 
in Section~\ref{charzero}. Finally, we study the uniqueness of 
optimal representations in Section~\ref{uniquesec}.
It turns out that the algorithms of Sections~\ref{distinctsec} 
and~\ref{repeatsec} only work in a regime where the uniqueness of optimal
representations is guaranteed.

\subsection{The real case} \label{realsec}

In \cite{GK} the authors considered polynomials with real coefficients and proved the following result.
\begin{theorem}\cite[Theorem 13]{GK} \label{lbtool}
Consider a polynomial identity of the form:
	\[
		\sum_{i=1}^k \alpha_i (x-a_i)^d = \sum_{i=1}^l \beta_i (x-b_i)^{e_i}
	\]
where the $a_i \in \RR$ are distinct constants, the constants
$\alpha_i \in \RR$ are not all zero, the $\beta_i \in \RR$ and $b_i \in \RR$ are arbitrary
 constants, and $e_i < d$ for every $i$. Then, we must have $k+l
\geq \lceil (d+3)/2 \rceil$.
\end{theorem}
Theorem~\ref{lbtool} will be our main tool in Section~\ref{realsec}.
As a consequence of this result,
we first give a sufficient condition for a polynomial 
to have a unique optimal expression in the model $\AffPow_\RR$.
\begin{corollary} \label{realuniquenessExpression}
	Let $f \in \RR[x]$ be a polynomial of the form:
	\begin{equation} \label{uniquereal}
		f = \sum_{i=1}^s \alpha_i (x-a_i)^{e_i}
	\end{equation}
	with $\alpha_i \not= 0$.
	For every $e \in \NN$ we denote by $n_e$ the number of exponents smaller than $e$, i.e., $n_e = \#\{i : e_i \leq e\}$.\\
	If $2n_e \leq \lceil (e+3)/2 \rceil$ for all $e \in \NN$, then $\AffPow_\RR(f) = s$.
	Moreover, if $2n_e <  \lceil (e+3)/2 \rceil$ for all $e$ 
then (\ref{uniquereal}) is the unique optimal expression for $f$.
\end{corollary}
\begin{proof}
	Suppose that $f$ can be written in another way
	\begin{equation} \label{uniquereal2}
		f = \sum_{j=1}^p \beta_j (x - b_j)^{f_j}
	\end{equation}
	with $p \leq s$.
	Set $d = \max\left( (e_i)_{1\leq i \leq s} \cup (f_j)_{1\leq j \leq p} \right)$ and denote by $s'$ (respectively, $p'$) the index such that $d = e_1 = \dots = e_{s'} > e_{s' + 1} \geq \dots \geq e_s$ (respectively, $d = f_1 = \dots = f_{p'} > f_{p' + 1} \geq \dots \geq f_p$). Note that one of the two indices $s',p'$ 
will be equal to 0 if the exponent $d$ appears only in one of the two expressions~(\ref{uniquereal}) and~(\ref{uniquereal2}).

	Combining equations (\ref{uniquereal}) and (\ref{uniquereal2}), we obtain the following equality:
	\[
		\sum_{i=1}^{s'} \alpha_i (x-a_i)^d - \sum_{j=1}^{p'} \beta_j (x-b_j)^d
		=
		-\sum_{i=s'+1}^{s} \alpha_i (x-a_i)^{e_i} + \sum_{j=p'+1}^{p} \beta_j (x-b_j)^{f_j}
	\]
	We can rewrite this as
	\[
		\sum_{i=1}^{k} \alpha'_i (x-a'_i)^d
		=
		\sum_{i=1}^{l} \beta'_i (x-b'_i)^{e'_i}
	\]
	with $\alpha_i' \not= 0$, $k \leq s' + p'$ and $l \leq (s - s') + (p - p')$. \\
To prove the first assertion, let us assume that 
$2n_e \leq \lceil (e+3)/2 \rceil$ for all $e$.
	Assume also for contradiction that $p < s$ and $k > 0$.
	By Theorem \ref{lbtool}, we must have $k + l \geq \lceil (d+3)/2 \rceil$.
	The upper bounds on $k$ and $l$ imply $2s > s + p \geq k+l \geq \lceil (d+3)/2 \rceil$.
	However we have from our assumption that $2s = 2n_d \leq 2\lceil (d+3)/2 \rceil$, which contradicts the previous inequality.
	This shows that $p < s \Rightarrow k = 0$, i.e., if $p < s$ then the highest degree terms are the same.
Continuing by induction, we find that all the terms in the two expressions are
equal. In particular we would have $p = s$, a contradiction.
	This shows that $p \geq s$, i.e., that $\AffPow_\RR(f) = s$.
	
To prove the second assertion, 
let us now assume further that $2n_e < \lceil (e+3)/2 \rceil$ for all $e$.
	Assume also that $p = s$.
	By Theorem \ref{lbtool}, either $k = 0$
or $k+l \geq \lceil (d+3)/2 \rceil$.
In the second case,
the upper bounds on $k$ and $l$ imply that $2s = s + p \geq k+l \geq \lceil (d+3)/2 \rceil$.
This is in contradiction with the assumption that 
$2n_d < \lceil (d+3)/2 \rceil$.
We conclude that that $k$ must be equal to 0, i.e., 
the highest degree terms are the same.
Continuing by induction, we obtain that all the terms of the two decompositions are equal, thus showing that~(\ref{uniquereal})
is the  unique optimal expression for $f$ in this model.
\end{proof}

 Let $\mathbb K$ be a field extension of $\FF$. Theorem 1 in \cite{LS} shows that whenever the value $\Sparsest_{\mathbb K}(f)$ is "small", then it is equal to
$\Sparsest_\FF(f)$; more precisely, 
if $\Sparsest_{\mathbb K}(f) \leq (d+1)/2$ then $\Sparsest_{\mathbb K}(f) = \Sparsest_\FF(f)$. This is no longer the case for the Affine Power model as the following
example shows.
 
\begin{example}\label{realvscomplex}
 	For every $d \in \NN$, we consider the polynomial 
 	\begin{equation}\label{complexandreal}
 		f_d := \sum_{j \equiv 3 \, ({\rm mod}\, 4) \atop 0 \leq j \leq d} 4 \binom{d}{j} x^{d-j} \in \RR[x].
 	\end{equation}
 	We can express $f_d$ as $f_d = (x+1)^d - (x-1)^d + i (x+i)^d - i (x-i)^d$, which proves that $\AffPow_{\CC}(f_d) \leq 4$. 
Moreover, in expression~(\ref{complexandreal}) we have $n_e \leq \lceil (e+1)/4 \rceil$ for all $e \in \NN$. 
Since $2\lceil (e+1)/4 \rceil \leq \lceil (e+3)/2 \rceil$,
it follows from Corollary~\ref{realuniquenessExpression} 
that this expression for $f_d$ is optimal over the reals, i.e., 
 	$\AffPow_{\RR}(f_d) = \lfloor (d+1)/4 \rfloor$.
\end{example}

As a consequence of Theorem \ref{lbtool} we can easily derive the following 
result.
\begin{corollary}\label{sparsepluswaringReal}
	Let $f \in \RR[x]$ be a polynomial of degree $d$.
	Either $f = \alpha (x-a)^d$ for some $\alpha, a \in \RR$ (and $ \Waring_\RR(f) = \Sparsest_\RR(f) = 1$),
	or the following holds:
	\[
		\Waring_\RR(f) + \Sparsest_\RR(f)  \geq \frac{d+3}{2}
	\]
\end{corollary}
\begin{proof}
	We set $k = \Waring_\RR(f)$  and $l = \Sparsest_\RR(f)$ and assume that $l \geq 2$. We write $f$ in two different ways:
		$$ f = \sum_{i=1}^{k}  \alpha_i (x - a_i)^d=  \sum_{j=1}^{l} \beta_i (x - a)^{e_i},$$
	where the $a_j \in \RR$ are all distinct, 
and $e_1 < \cdots < e_{l} = d$.
Let us  move the term $\beta_l (x-a)^d$ to the left
	hand side of the equation. We then have two cases to consider:
	\begin{itemize}
		\item if $a \neq a_i$ for all $i$, we have $k+1$ terms on the left hand side of the equation and $l-1$ terms on the right hand side.
Theorem \ref{lbtool} shows that 
$(k+1) + (l-1) \geq (d+3)/2$.
		
\item If $a = a_i$ for some $i$, we have 
$k$ or
$k-1$ terms on the left hand side of the equation  
and $l-1$ terms on the right hand side.
By Theorem \ref{lbtool}, 
$k + (l-1) \geq (d+3)/2$.
	\end{itemize}
\end{proof}

\begin{remark} Consider the degree $d \geq 2$ polynomial 
$$f := (x+1)^d  + (x-1)^d = \sum_{i {\text \ even} \atop 0 \leq i \leq d} 2 \binom{d}{i} x^{d-i}.$$
We observe that $\Waring_\RR(f) = 2$ and $\Sparsest_\RR(f) \leq \lceil (d+1)/2 \rceil$. Hence, the inequality in Corollary \ref{sparsepluswaringReal} is optimal up to one
unit. 
\end{remark}

A similar proofs to that of Corollary \ref{sparsepluswaringReal} yield the following result:

\begin{corollary}\label{affpowpluswaringReal}
	Let $f \in \RR[x]$ be a polynomial of degree $d$.
	Either $\AffPow_\RR(f) = \Waring_\RR(f)$  
	or the following inequality holds:
	\[
		\Waring_\RR(f) + \AffPow_\RR(f) \geq \frac{d+3}{2}
	\]
\end{corollary}

\subsection{Fields of characteristic zero} \label{charzero}

We now switch from the real field to an arbitrary field $\FF$ 
of characteristic zero.
By definition we have $\AffPow_{\FF}(f) \leq \Waring_{\FF}(f)$ and $\AffPow_\FF(f) \leq \Sparsest_\FF(f)$ for any polynomial $f \in {\FF}[X]$. 
We show in Example~\ref{affpowsmaller} that there are polynomials $f$
such that $\AffPow_\FF(f)$ is much smaller than both 
$\Waring_\FF(f)$ and  $\Sparsest_\FF(f)$.

We first make some basic observations about Sparsest Shift. 
For any $a \in \FF$,  the polynomials $\{(x-a)^i\, \vert \, i \in \mathbb{N}\}$ are linearly independent, hence $f$ can be uniquely 
expressed as $f = \sum_{i = 0}^d \alpha_i (x-a)^i$ 
where $\alpha_i = f^{(i)}(a)/i!$.
Consider such a decomposition for $f$, and let $s$ be the number of nonzero terms. It follows that the $d+1-s$ derivatives $f^{(i)}$ with $\alpha_i=0$ admit
$a$ as a common root.
\begin{example} \label{affpowsmaller}
 	For every $d \in \mathbb{N}$, we consider the polynomial $f_d := (x+1)^d - d x^{d-1} \in \mathbb{C}[x]$.
 	It is easy to check that $\AffPow(f_d) = 2$ for all $d \geq 2$. By \cite[Proposition 3.1]{BCG} we have that if $x^{d-1} = \sum_{i = 1}^s \alpha_i (x-a_i)^d$
 	with $\alpha_i, a_i \in \mathbb{C}$, then $s \geq d$; and thus we get that $\Waring_{\mathbb{C}}(f_d) \geq d-1$.
 	
One can easily check that for every $i \in \{0,\ldots,d-1\}$, the polynomials $f_d^{(i)} = \frac{d!}{(d-i)!} f_{d-i}$ and $f_d^{(i+1)} = \frac{d!}{(d-i-1)!} f_{d-i-1}$ do not share a common root. Consider a decomposition of $f$ in the sparsest shift model. By the above observations, for any pair of consecutive coefficients in this decomposition at least one of the 2 coefficients is nonzero.
This implies 
that $\Sparsest_\mathbb{C}(f) \geq \lceil (d+1)/2 \rceil$.
\end{example}

In the remainder of Section~\ref{charzero} we give (in Proposition~\ref{sparsetimeswaringChar0}) a weaker version of Corollary \ref{sparsepluswaringReal} that works for any field of characteristic zero. Moreover, for $\FF = \CC$ we  
provide a family of polynomials showing that the bound from Proposition~\ref{sparsetimeswaringChar0}  is sharp. 

We will use Jordan's lemma \cite{GrY} 
(see \cite[Lemma 1.35]{IaKa} for a recent reference), 
which can be restated as follows.

\begin{lemma}\label{JordanLemma}Let $d \in \ZZ^+$, $e_1,\ldots,e_t \in \{1,\ldots,d\}$, and let $a_1,\ldots,a_t \in \FF$ be distinct constants. If $\,\sum_{i = 1}^t (d+1-e_i) \leq d+1$, then the set of
polynomials 
$$\{ (x-a_i)^{e} \, \vert \, 1 \leq i \leq t,\, e_i \leq e \leq d\}$$ is linearly independent.
\end{lemma}	

\begin{proposition}\label{sparsetimeswaringChar0}
	Let $f \in \FF[x]$ be a polynomial of degree $d$.
	Either $f = \alpha (x-a)^d$ for some $\alpha, a \in \FF$ (and $\Waring_\FF(f) = \Sparsest_\FF(f) = 1$),
	or the following holds:
	\[
		\Waring_\FF(f) \cdot  \Sparsest_\FF(f) \geq d + 1
	\]
\end{proposition}
\begin{proof}
	We set $k = \Waring_\FF(f)$ and $l = \Sparsest_\FF(f)$ and assume that $k,l \geq 2$. We express $f$ in two different ways:
		$$ f = \sum_{i=1}^{k} \alpha_i (x - a_i)^d =  \sum_{j=1}^{l} \beta_j (x - a)^{e_j},$$
	with $a_j \in \FF$ all distinct and $e_0 := -1 < e_1 < \cdots < e_l = d$. First, we are going to prove that $e_{i+1} - e_i \leq k$ for all 
	$i \in \{0,\ldots,l-1\}$. Indeed, if there exists $t \in \{0,\ldots,l-1\}$ such that $e_{t+1} - e_t \geq k + 1$, then we set $r := e_{t} + 1$ and differentiate
	the previous equality $r$ times to obtain 
		$$ f^{(r)} =  \sum_{i=1}^{k} \alpha_i \frac{d!}{(d-r)!} (x - a_i)^{d-r} = \sum_{j=t+1}^{l} \beta_j \frac{e_j !}{ (e_j - r)! } (x - a)^{e_j-r},$$
	where $e_j - r = e_j - e_t - 1 \geq e_{t+1} - e_t - 1 \geq k$ for all $j \in \{t+1,\ldots,l\}$. 
	From this equality, we deduce that the set 
		$$\mathcal B := \{(x-a_i)^{d-r} \, \vert \,1 \leq i \leq k\} \cup \{(x-a)^{e_i - r} \, \vert \, t+1 \leq i \leq l\}$$
	is linearly dependent. However,  
		$$\mathcal B \subseteq  \{(x-a_i)^{d-r} \, \vert \,1 \leq i \leq k\} \cup \{(x-a)^{i} \, \vert \, k \leq i \leq d-r\}.$$
The $d - r + 1$ polynomials on the right-hand side are of degree at most $d-r$, and they are linearly independent by Jordan's lemma.
This is a contradiction since $\mathcal B$ is linearly dependent. 
 We have proved that $e_{i+1} - e_i \leq k$ for all $i \in \{0,\ldots,l-1\}$, 
and we conclude that
		$$d + 1 = e_l - e_0 = \sum_{i = 1}^{l} (e_{i} - e_{i-1}) \leq k l.$$
\end{proof}

\begin{remark}\label{tightproduct}
	One can slightly modify \cite[Proposition 19]{GK} to obtain the following equality of complex polynomials of degree $d$:
	\[
		\sum_{j=1}^k (x + \xi^j)^d = \sum_{0 \leq i \leq d \atop i \equiv 0 \ (\textup{mod}\ k)} k \binom{d}{i} x^{d-i}
	\]
	where $k \in \NN$ and $\xi \in \CC$ is a $k$-th primitive root of unity. This equality shows that there are polynomials of degree $d$ 
such that $\Waring_\CC(g) \leq k$ and
	$\Sparsest_\CC(g) \leq  \lceil (d+1)/ k \rceil$ and, thus the bound from Proposition \ref{sparsetimeswaringChar0} is tight. \\

\end{remark}


 \subsection{Uniqueness results for sums of affine powers} \label{uniquesec}

The following result is an analogue of Theorem \ref{lbtool} for polynomials with coefficients over $\FF$, where $\FF$ is any field of characteristic zero.

\begin{proposition}\label{LowerBound}
	Consider a polynomial identity of the form:
	\[
		\sum_{i = 1}^k \alpha_i (x - a_i)^d = \sum_{i = 1}^l \beta_i (x - b_i)^{e_i}
	\]
	where the $a_i \in \FF$ are distinct, the $\alpha_i \in \FF$ are not all zero, $\beta_i, b_i \in \FF$ are arbitrary, and $e_i < d$ for every $i$.
	Then we must have $k+l > \sqrt{2(d+1)}$.
\end{proposition}
\begin{proof}
	We assume $\alpha_1 \not= 0$ and we have the following equality:
	\[
		\alpha_1 (x - a_1)^d = - \sum_{i = 2}^k \alpha_i (x - a_i)^d + \sum_{i = 1}^l \beta_i (x - b_i)^{e_i}
	\]
	Consider an independent subfamily on the right hand side of this equality. We obtain a new identity of the form:
	\[
		g = \sum_{i=1}^p \lambda_i \ell_i^{\, r_i}
	\]
	with $g(x) = \alpha_1 (x-a_1)^d$, and $p \leq k + l - 1$. 
	Since ${\rm deg}(g) = d$ and $e_i < d$ for all $i$; then there exists $i$ such that $r_i = d$. We assume without loss of generality that $\ell_1 = x - a_2$ and $r_1 = d$.
	
	We take the derivatives of this equality to obtain the following system:
	\begin{align*}
		g &= \sum_{i=1}^p \lambda_i \ell_i^{\, r_i} \\
		g^{'} &= \sum_{i=1}^p \lambda_i \left[ \ell_i^{\, r_i} \right]^{'}\\
		\vdots  &  \\
		g^{(p-1)} &= \sum_{i=1}^p \lambda_i \left[ \ell_i^{\, r_i} \right]^{(p-1)}\\
	\end{align*}
	
	Using Cramer's rule, we obtain:
	\[
		\lambda_1 = \frac{\Wr(g, \ell_2^{\, r_2}, \dots, \ell_p^{\, r_p})}{\Wr(\ell_1^{\, r_1},\ell_2^{\, r_2}, \dots, \ell_p^{\, r_p})}
	\]
	
	We define $\Delta = \{i : 2 \leq i \leq p, r_i \geq p\}$ and, following Proposition \ref{factorWronskian}, we factorise the Wronskians:
	\[
		\lambda_1 = \frac{(x-a_1)^{d - (p-1)} \prod_{i \in \Delta} \ell_i^{\, r_i - (p-1)} \cdot W_1}{(x-a_2)^{d - (p-1)} \prod_{i \in \Delta} \ell_i^{\, r_i - (p-1)} \cdot W_2}
	\]
	where $W_1,W_2$ are the remaining determinants.
	
	\noindent	
	After some simplifications, we obtain the following identity:
	\[
		\lambda_1 (x-a_2)^{d - (p-1)} W_2 = (x-a_1)^{d - (p-1)} W_1
	\]
	Notice now that since we have factorised the large $r_i$'s, the $i^{th}$ row of $W_1$ and $W_2$ contains polynomials with degree bounded by $p - i$, thus $\deg W_1, \deg W_2 \leq p(p-1)/2$.
	
	Moreover, since $a_1 \not= a_2$, we compute the multiplicity of $a_1$ on both sides of the identity and obtain that
	\[
		\mult[a_1]{(x-a_1)^{d - (p-1)} W_1} = \mult[a_1]{\lambda_1 (x-a_2)^{d - (p-1)} W_2} = \mult[a_1]{W_2}.
	\]
	The previous remark on the degree of $W_2$ therefore implies that
	\[
		d - (p-1) \leq \frac{p(p-1)}{2}
	\]
	Finally, we set $s = l + k$ and we use the fact that $p \leq s - 1$ to obtain the desired lower bound:
	\begin{align*}
		d  &\leq \frac{(p+2)(p-1)}{2} \cr
		d  &\leq \frac{(s+1)(s-2)}{2} \cr
		2d &\leq s^2 - s - 2 \cr
	\end{align*}
and finally, $2(d+1) < s^2$.
\end{proof}

\begin{remark}
	The same equality as in Remark \ref{tightproduct} shows that the order of this bound is tight when $\FF = \CC$, the field of complex numbers.
	Indeed, choosing $k = \sqrt{d+1}$ leads to the equality
	\[
		\sum_{i=1}^{k} (x + \xi^i)^d = \sum_{j = 0}^{k-1} k \binom{d}{jk} x^{d-jk}
	\]
	which has $2k = 2\sqrt{d+1}$ terms.
\end{remark}

As a consequence of Proposition \ref{LowerBound} we obtain that whenever 
$\AffPow_\FF(f)$ is sufficiently small, the
terms of highest degree in an optimal expression of $f$ as $f = \sum_{i = 1}^s \alpha_i (x-a_i)^{e_i}$ are uniquely determined.

\begin{corollary}\label{uniquenessHighestDegree}
	Let $f \in \FF[x]$ be a polynomial of the form :
	\[
		f = \sum_{i=1}^k \alpha_i (x-a_i)^d + \sum_{j=1}^l \beta_j (x-b_j)^{e_j}
	\]
	with $e_j < d$.
	If $k + l \leq \sqrt{\frac{d+1}{2}}$, then the highest degree terms are unique. In other words, for every expression of $f$ 
	as
	\[
		f = \sum_{i=1}^{k'} \alpha_i' (x-a_i')^d + \sum_{j=1}^{l'} \beta_j' (x-b_j')^{e_j'}
	\]
	with $e_j' < d$ and $k' + l' \leq \sqrt{\frac{d+1}{2}}$, then $k = k'$ and there exists a permutation $\pi: \{1,\ldots,k\} \rightarrow \{1,\ldots,k\}$ such that
	$\alpha_i = \alpha_{\pi(i)}'$ and $a_i = a_{\pi(i)}'$ for all $i \in \{1,\ldots,k\}$.
\end{corollary}
\begin{proof}
	Let us assume that we have another different decomposition for $f$:
	\[
		f = \sum_{i=1}^{k'} \alpha_i' (x-a_i')^d + \sum_{j=1}^{l'} \beta_j' (x-b_j')^{e_j'}
	\]
	with $k' + l' \leq \sqrt{(d+1)/2}$. Hence, we have the following equality:
	\[
		\sum_{i=1}^{k} \alpha_{i} (x-a_i)^d -  \sum_{i=1}^{k'} \alpha_{i}' (x-a_i')^d
		= \sum_{j=1}^{l} \beta_{j} (x-b_j)^{e_{j}} - \sum_{j=1}^{l'} \beta_{j}' (x-b_j')^{e_{j}'}
	\]
	Since $k + k' + l + l' \leq \sqrt{2(d+1)}$, the result follows from Proposition~\ref{LowerBound}.
\end{proof}

\medskip

Finally, as a direct consequence of  Corollary \ref{uniquenessHighestDegree}, 
we obtain a
a sufficient condition for a polynomial to have a unique optimal expression in the $\AffPow$ model:
\begin{corollary} \label{uniquenessExpression}
	Let $f \in \FF[x]$ be a polynomial of the form:
	\[
		f = \sum_{i=1}^s \alpha_i (x-a_i)^{e_i}
	\]
	For every $e \in \NN$ we denote by $n_e$ the number of exponents smaller than $e$, i.e., $n_e = \#\{i : e_i \leq e\}$.
	If $n_e \leq \sqrt{\frac{e+1}{2}}$ for all $e \in \NN$, then $\AffPow_\FF(f) = s$ and the optimal representation of
	$f$ is unique.
\end{corollary}

\begin{remark}
	Whenever $f \in \RR[x]$ satisfies the hypotheses of Corollary \ref{uniquenessExpression} and one term in the expression of $f$ is of the form
	$\alpha_i (x - a_i)^{e_i}$ with $a_i \in \CC-\RR$, then there exists $j \not= i$ such 
	that $\alpha_j = \overline{\alpha_i}, a_j = \overline{a_i}$ and $e_j = e_i$. 
	Indeed, if we have a decomposition for $f$, taking the conjugate of $\alpha_i$ and $a_i$ for all $i$ gives another decomposition of $f$, but by 
	Corollary \ref{uniquenessExpression} these two decompositions must be identical. In Proposition \ref{Galois} we will prove a more general version of this fact.
	
\end{remark}

Another consequence of Proposition \ref{LowerBound} is the following upper bound on the degree of the terms involved in an optimal expression of $f$ 
in the model $\AffPow_{\FF}$.

\begin{corollary}\label{upperbounddegreeterms} Let $f \in \FF[x]$ be a polynomial of degree $d$ written as
	\[
		f = \sum_{i = 1}^s \alpha_i (x-a_i)^{e_i}
	\]
	with $\alpha_i, a_i \in \FF$, $e_i \in \mathbb{N}$ and $s = \AffPow_\FF(f)$. We set $e := {\rm max}\{e_i \, : \, 1 \leq i \leq s\}$, then
	$e < d + \frac{s^2}{2}$ and, if $\, \FF = \mathbb{R}$, then $e \leq d + 2s - 2$. In particular, 
	we have that $e < d + \frac{(d+2)^2}{8}$ and, if $\, \FF = \mathbb{R}$, then $e \leq 2d$.
\end{corollary}
\begin{proof} 
	If $e = d$, then the result is trivial. Assume therefore that $e > d$. 
 Now,
	we differentiate $d+1$ times the expression for $f$ 
to obtain the identity:
	\[
		0 = f^{\,(d+1)} = \sum_{e_i > d} \alpha_i \frac{e_i!}{(e_i - d - 1)!} (x-a_i)^{e_i-d-1}.
	\]
	By Proposition \ref{LowerBound} we have $s > \sqrt{2(e - d)}$ and we conclude that $e < d + \frac{s^2}{2}$.
	When $\FF = \mathbb{R}$, by Theorem \ref{lbtool} we have $s \geq (e-d+2)/2$ and we conclude that $e \leq d + 2s - 2$. 
	To finish the proof it suffices to recall that $s = \AffPow_{\FF}(f) \leq \lceil (d+1)/2 \rceil \leq (d+2)/2$; see \cite[Proposition 18]{GK}.
\end{proof} 

\bigskip

\begin{remark}
	On can find examples that are close to the bound of Corollary \ref{upperbounddegreeterms}. Indeed, if we take $k = \sqrt{d + 1}$ in Remark \ref{tightproduct}, we get an expression of the $0$ polynomial with $2k$ terms, namely:
	\[		
		\sum_{j=1}^k (x + \xi^j)^d - \sum_{0 \leq i \leq d \atop i \equiv 0 \ (\textup{mod}\ k)} k \binom{d}{i} x^{d-i} = 0
	\]
	 If we integrate this expression $7(d+1)$ times we get a polynomial 
	 \[		
		f := \sum_{j=1}^k (x + \xi^j)^{8d+7} - \sum_{0 \leq i \leq d \atop i \equiv 0 \ (\textup{mod}\ k)} k \binom{d}{i} x^{8d+7-i},
	\]
	of degree $< 7(d+1)$ with $s := \AffPow_\FF(f) = 2k$ (by Corollary \ref{uniquenessExpression}) and whose  maximum exponent in the optimal expression is $8d + 7 = 7(d+1) + d < {\rm deg}(f) + (s^2-4)/4$.  
\end{remark}

\begin{remark}
	As a consequence of Corollary \ref{upperbounddegreeterms}, we obtain a naive brute force algorithm to find one optimal expression for any polynomial $f$.
	Indeed, for a fixed integer $s$, there are only a finite number of sequences of exponents $(e_1, \dots, e_s)$ with $e_i \leq d + s^2/2$.
	For one sequence, one can try to find an expression with these exponents by solving a system of polynomial equations in $2s$ variables.
	The smallest $s$ with a solution gives the value of $\AffPow_\FF(f)$.
\end{remark}

\bigskip

Also, as a byproduct of Corollary \ref{upperbounddegreeterms}, we obtain the 
exact value of $\AffPow_{\FF} (f)$ for a generic polynomial $f$ of degree $d$.
It turns out to be equal to the worst case value of $\AffPow_{\FF} (f)$,
obtained in \cite[Proposition 18]{GK}.
\begin{corollary}For a generic polynomial $f \in \FF[x]$ of degree $d$, 
 $\AffPow_{\FF}(f) = \lceil \frac{d+1}{2} \rceil$.
\end{corollary}
\begin{proof} 
	The set of polynomials of degree $\leq d$ can be seen as a variety $W$ of dimension $d+1$. 
	Given $f \in \FF[x]$ a polynomial of degree $d$, by \cite[Proposition 18]{GK} we have $\AffPow_{\FF} (f) \leq \lceil \frac{d+1}{2} \rceil$.
	For $k < \lceil \frac{d+1}{2} \rceil$, let us show that the set of polynomials $g$ of degree $d$ such that $\AffPow_{\FF}(g) \leq k$ is
	contained in a variety of dimension $2k < d+1$. For every $e_1,\ldots,e_k \in \mathbb{N}$ the set of polynomials that can be written 
	as $\sum_{i = 1}^{k} \alpha_i (x - a_i)^{e_i}$ with $a_i, \alpha_i \in \FF$
	is contained in a variety $V_{e_1,\ldots,e_k}$ of dimension $2k$. If we set $M := d + \frac{(d+2)^2}{8}$, Corollary \ref{upperbounddegreeterms} proves
	that in every optimal expression of a polynomial of degree $d$, the exponents $e_i$ are $\leq M$; thus the set of polynomials with $\AffPow_{\FF}(f) \leq k$
	and degree $d$ is contained in $\bigcup_{e_i \leq M} V_{e_1,\ldots,e_k}$, which	is a variety of dimesion $\leq 2k$ (it is a 
	finite union of varieties of dimension $\leq 2k$).
\end{proof} 

\section{Algorithms for distinct nodes} \label{distinctsec}

The goal of this and the following section is to provide algorithms that receive as input a polynomial $f$ and computes 
$s = \AffPow_\FF(f)$ and the triplets $(\alpha_i,a_i,e_i)$ for $i \in \{1,\ldots,s\}$ such that $f = \sum_{i = 1}^s \alpha_i (x-a_i)^{e_i}$.
We will not able to solve the problem in all its generality but under certain hypotheses. This section concerns the case where the $a_i$ in
the optimal expression of $f$ are all distinct. In this setting, our main result is Theorem \ref{diffnodesimproved} where we
solve the problem when the number $n_e$ of exponents in the optimal expression 
that are $\leq e$ is 'small'. A key point to obtain the algorithms is given by the following Proposition. Roughly speaking, this result says that if $f$ satisfies
a $\SDE$, then every term in the optimal expression of $f$ with exponent $e_i$ big enough also satisfies the same $\SDE$.

\begin{proposition}\label{bigexponents}
	Let $f \in \FF[x]$ be written as
	\[
		f = \sum_{i = 1}^s \alpha_i (x - a_i)^{e_i},
	\]
	with $\alpha_i \in \FF$ nonzero, the $a_i \in \FF$ are all distinct, and $e_i \in \NN$.
	Whenever $f$ satisfies a $\SDE(k,l)$, then for all $e_i \geq k + (k+l)(s-1) + \binom{s}{2}$ 
	we have that $(x-a_i)^{e_i}$ satisfies the same SDE.
\end{proposition}
\begin{proof}
	We assume that  $e_1 \geq k + (k+l)(s-1) + \binom{s}{2}$ and that $f$ satisfies the following $\SDE(k,l)$:
	\[
		\sum_{i = 0}^k P_i(x)\, g^{(i)}(x) = 0,
	\]
	with $\deg(P_i) \leq i + l$.
	By contradiction, we assume that $(x-a_1)^{e_1}$ does not satisfy this equation.
	For every $j \in \{1,\ldots,s\}$,  we denote by $f_j$ and $R_j$ the polynomials such that
	\[
		f_j = \sum_{i = 0}^k P_i(x)\, ((x-a_j)^{e_j})^{(i)} = R_j(x)\, (x-a_j)^{d_j},
	\]
	where $d_j := \max\{e_j-k,0\}$.
 	We observe that  $\deg(f_j) \leq e_{j} + l $, so $\deg(R_j) \leq k + l$, and that $-f_1 = \sum_{j = 2}^s f_j \neq 0$.
	We consider a linearly independent subfamily of $f_2,\ldots,f_s$, namely $\{f_j \, \vert \, j \in J\}$ with
	$J = \{j_1,\ldots,j_p\} \subseteq \{2,\ldots,s\}$. 
	Then by Proposition \ref{factorWronskian} we have that
	\[
		\begin{array}{cl} e_1-k = d_1  \leq \mult[a_1]{f_1} & \leq p - 1 + \sum_{j \in J}{\rm deg}(R_j) + (p-1)p - \binom{p}{2} \\ & \leq p-1 + (k+l)p + \binom{p}{2}. \end{array}
	\] 
	Since $p \leq s-1$, we get that $e_1 \leq k + s - 2 + (k+l)(s-1) + \binom{s-1}{2}  < k + (k+l)(s-1) + \binom{s}{2}$, a contradiction. 
\end{proof}

\bigskip


\medskip

As a consequence of Proposition \ref{bigexponents}, we get Corollary \ref{bigexpCor} and Theorem \ref{th-algbigexp}. They provide an effective
method to obtain the optimal expression of a polynomial $f$ in the Affine Power model whenever all the terms involved have big exponents and all the nodes
are different.

\begin{corollary}\label{bigexpCor}
	Let $f \in \FF[x]$ be written as
		$f = \sum_{i = 1}^s \alpha_i (x - a_i)^{e_i},$ 
	with $\alpha_i \in \FF \setminus \{0\}$, $a_i \in \FF$ all distinct, and $e_i \geq 5 s^2/2\ $ for all $i$.
	Then, 
	\begin{enumerate}[\quad a)]
		\item $\{(x-a_i)^{e_i}\, \vert \, 1 \leq i \leq s\}$ are linearly independent,

		\item If $f = \sum_{i = 1}^t \beta_i (x - b_i)^{d_i}$ with $t \leq s$, then $t = s$ and we have the equality
	$\{(\alpha_i,a_i,e_i) \, \vert \, 1 \leq i \leq s\} = \{(\beta_i,b_i,d_i) \, \vert \, 1 \leq i \leq s\}$; in particular,
		$\AffPow_\FF(f) = s$,
	
		\item $f$ satisfies a $\SDE(2s-1,0)$,
		
		\item if $f$ satisfies a $\SDE(k,0)$ with $k \leq 2s-1$ then $(x-a_i)^{e_i}$ also satisfies it for all $i \in \{1,\ldots,s\}$, and
		
		\item $f$ does not satisfy any $\SDE(k,0)$ with $k < s$.
	\end{enumerate}
\end{corollary}
\begin{proof}
	Notice first that (b) implies (a). Assume now that (b) does not hold, then there is another expression of $f$ as $f = \sum_{i = 1}^t \beta_i (x-b_i)^{d_i}$ with $t \leq s$. Hence, 
	by Proposition \ref{LowerBound}, we get that 
	\[
		2s \geq t + s > \sqrt{2 (\min(\{e_1,\ldots,e_s\}) + 1)} \geq \sqrt{5s^2},
	\]
	a contradiction. From Proposition \ref{smallequation} we get (c). If $f$  satisfies a $\SDE(k,0)$ with $k \leq 2s-1$,
	then for all $i \in \{1,\ldots,s\}$ we have that
	\[
		e_i \geq 5s^2 / 2 \geq (2s-1)s + \binom{s}{2} \geq k s + \binom{s}{2}.
	\]
	Hence, Proposition \ref{bigexponents} yields that $(x-a_i)^{e_i}$ is also a solution of this
	equation for all $i$, proving (d). Finally, $f$ cannot satisfy a $\SDE(k,0)$ with $k < s$; otherwise by (a) and (d), the vector space of solutions to this equation 
	has  dimension $\geq s$, which contradicts Lemma \ref{dimsolution}.
\end{proof}

\begin{th-algorithm}[Big exponents]\label{th-algbigexp}
	Let $f \in \FF[x]$ be a polynomial that can be written as 
	\[
		f = \sum_{i = 1}^s \alpha_i (x - a_i)^{e_i},
	\]
 	where the constants $a_i \in \FF$ are all distinct, $\alpha_i \in \FF \setminus \{0\}$ and $e_i > 5 s^2/2$. 
 	Then, $\AffPow_\FF(f) = s$. Moreover, there is a polynomial time algorithm {\tt Build}$(f)$ that receives $f = \sum_{i = 0}^d f_i x^i \in \FF[x]$ as input and computes 
 	the $s$-tuples of coefficients  $C(f) = (\alpha_1,\ldots,\alpha_s)$, of nodes $N(f) = (a_1,\ldots,a_s)$ and exponents $E(f) = (e_1,\ldots,e_s)$. 
	The algorithm {\tt Build}$(f)$ works as follows:

	\begin{enumerate}[\; \bf Step 1.]
		\item Take $r$ the minimum value such that $f$ satisfies a $\SDE(r, 0)$ and compute explicitly one of these $\SDE$.

		\item Compute $B = \{(x-b_i)^{d_i}\, \vert \, 1 \leq i \leq l\}$, the set of all the solutions of the $\SDE$ of the form $(x-b)^e$
		with $(r+1)^2/2 \leq e \leq {\rm deg}(f) + (r^2/2)$. 

		\item Determine $\beta_1,\ldots,\beta_{l}$ such that $f = \sum_{i = 1}^{l}  \beta_i (x-b_i)^{d_i}$

		\item Set $I := \{i \, \vert \,\beta_i \neq 0\}$ and output the sets $C(f) = (\beta_i \,\vert \, i \in I), \, N(f) = (b_i \,\vert \, i \in I)$
		and $E(f) = (d_i\,\vert \, i \in I)$. 
	\end{enumerate}
\end{th-algorithm}
\begin{proof}
	Corollary \ref{bigexpCor}  proves the correctness of this algorithm.
	Indeed, by Corollary \ref{bigexpCor}.(c) and (e), the value $r$ computed in {\bf Step 1} satisfies that $s \leq r \leq 2s-1$. 
	We claim that the set $B$ computed in {\bf Step 2} satisfies that: 
	\begin{itemize}
		\item[(1)] it contains the set  $\{(x-a_i)^{e_i} \, \vert \, 1\leq i \leq s)\}$,  
		\item[(2)] it has at most $r$ elements, and 
		\item[(3)] all its elements are $\FF$-linearly independent.
	\end{itemize} 
	The first claim follows from Corollary \ref{bigexpCor}.(d), the fact that $(r+1)^2 / 2 \leq  (2s)^2/2 < 5s^2/2$, and from Corollary \ref{upperbounddegreeterms},
	since $e_i \leq {\rm deg}(f) + (s^2/2) \leq {\rm deg}(f) + (r^2/2)$ for all $i$. To prove the second claim assume that $B$ has more than $r$ elements, then we take
	$t_1,\ldots,t_{r+1} \in B$. To reach a contradiction, by Lemma \ref{dimsolution} it suffices to prove that $t_1,\ldots,t_{r+1}$ are linearly independent.
	If this were not the case, by Proposition \ref{LowerBound}, we would have that $r+1 > \sqrt{(r+1)^2 + 2}$, which is not possible.
	A similar argument and the fact that $B$ has at most $r$ elements proves the third claim.
	By (1) and (3), the expression of $f$ as a combination of the elements of $B$ is unique and is the desired one.

	Finally, the four steps can be perfomed in polynomial time. Only the first
	two steps require a justification. 
	See Remark~\ref{findSDE} in Section~\ref{prelim} regarding Step~1.
	In Step~2 we substitute for each value of $e$ the polynomial $(x-b)^e$ in the
	SDE. This yields a polynomial $g(x)$ whose coefficients are polynomials
	in $b$ of degree at most $r \leq 2s-1$. We are looking for the values of $b$ which
	make $g$ identically 0, so we find $b$ as a root of the gcd of the coefficients 
	of $g$.
\end{proof}

\bigskip

  In the following result we are going to 
  analyze the bitsize complexity of the algorithm proposed; for this purpose we 
  assume that the output (and, hence, the input) have integer coefficients. 
  With this analysis we 
  intend to show a rough overestimate on the number of bitsize operations 
showing the polynomial time nature of the algorithm.

We recall that by the dense size of a polynomial $f = \sum_{i=0}^d f_ix^i \in \ZZ[X]$
we mean $size(f) := \sum_{i=0}^d [1+\log_2(1+|f_i|)]$. 
Also for an $n \times m$ matrix $M$ with rational entries $p_{i j}/q_{i j}$ where $p_{i j} \in \ZZ$, $q_{i j} \in \ZZ^+$, the  bit size  of $M$ is $size(M) := \sum_{i=1}^n \sum_{j = 1}^m [1+\log_2(1+|p_{i j}|)+\log_2(1+q_{i j})].$ 
 The notation $f(n) =  \overline{\mathcal O}(g(n))$ means that there exists
a $k \in \NN$ such that $f(n) = \mathcal O(g(n) \log^k(\max(|g(n)|,2)))$.

\begin{proposition}\label{runningtimebitsizeoperation}
Let $f$ be a polynomial of degree $d$ that can be written as 	
$$f = \sum_{i = 1}^s \alpha_i (x - a_i)^{e_i},$$ 
where the constants $a_i \in \ZZ$ are all distinct, $\alpha_i \in \ZZ \setminus \{0\}$ and $e_i > 5 s^2/2$. 
The algorithm ${\tt Build}(f)$ in Theorem \ref{th-algbigexp} outputs the optimal expression of $f$ in the $\AffPow$ model in $\mathcal{\mathcal O}(d^{\,6.5} size(f) + d^{\,8})$ bitsize operations.
\end{proposition}
\begin{demo}A first observation is that the value $r$ computed in {\bf Step 1} of the algorithm is upper bounded in terms of $d$. Indeed, by hypothesis $5s^2/2 \leq {\max}(e_i)$ and, by Corollary \ref{upperbounddegreeterms}, ${\max}(e_i) \leq d + (s^2/2)$, which implies that $d \geq 2s^2$. Moreover, in Corollary \ref{bigexpCor} we show that $r \leq 2s-1$; this gives $r = \mathcal O(\sqrt{d})$. Also by Corollary \ref{bigexpCor}, we have that $s \leq r$ and then ${\max}(e_i) = O(d)$.  

Let us study now the number of bitsize operations needed to obtain a $\SDE(r,0)$ satisfied by $f$ assuming that we know in advance the value of $r$ in {\bf Step 1} of the algorithm.
We propose to follow the idea of Remark \ref{findSDE} and find the $\SDE$ by computing a vector in the kernel of the matrix $M$
 whose entries are the coefficients of the polynomials $x^j f^{(i)}$ with $0 \leq j \leq i  \leq r$. We have that $M$ has $1 + \cdots + r = (r+1)r/2$ rows and $d+1$ columns.
  Since $size(x^j f^{(i)}) = \mathcal O(size(f) + i d \log(d))$, we have that  $size(M) =  \mathcal O(\sum_{i = 0}^r  (i+1)(size(f) + id\log(d))) = {\mathcal O}(r^2 (size(f) + r d \log(d)))$, which is $\overline{\mathcal O}(d\, size(f) + d^{\,2.5})$. 
  Now, we can obtain the required $\SDE$ by means of
the Gauss pivoting method on $M$. Let $E$ be the matrix in echelon form obtained by the  Gauss method. By \cite[Theorem 3.3]{Schrijver}, to compute $E$ one needs $\mathcal O(r^4 d)$  arithmetic operations, which is $O(d^{\,3})$, and the biggest size
of a coefficient appearing during the process of elimination by pivoting is $\mathcal O(size(M))$. Thus, the number of bitsize operations needed to obtain the $\SDE(r,0)$ is $ \overline{\mathcal O}(d^3\, size(M))$. Also, the biggest size of a coefficient appearing in the $\SDE(r,0)$ found is $\mathcal O(size(M))$.  After multiplying by an appropriate integer, we can assume that each of these coefficients are integers of size $\mathcal O(size(M))$.

We now lift the assumption that $r$ is known in advance.
To perform {\bf Step 1} we follow Remark \ref{findSDE} and
we check whether $f$ satisfies a $\SDE(\ell,0)$ starting from  $\ell = 0$
and increasing $\ell$.
We observe that at each step, we can check if $f$ satisfies a $\SDE(\ell,0)$ by checking if the matrix $M_{\ell}$ whose rows are the coefficients of the polynomials $x^i f^{(j)}$ with $0 \leq i \leq j  \leq \ell$ has full row rank. This can be easily checked from the matrix $E_{\ell}$ in echelon form obtained by applying the Gauss method to $M_{\ell}$. Since $M_{\ell}$ and $E_{\ell}$ are respectively submatrices of  $M_{\ell+1}$ and $E_{\ell+1}$,  the procedure of computing the 
$\SDE$ of smallest order satisfied by $f$ can be done incrementally. Moreover, all the matrices $M_{\ell}$ and $E_{\ell}$ are submatrices of the matrices $M$ and $E$ described above.  So, it is interesting to notice that knowing the exact value of $r$ in advance does not give any advantage and {\bf Step 1} can be performed in $\overline{\mathcal O}(d^{,3} size(M))$ bitsize operations.


To perform {\bf Step 2} we propose the following strategy. Assume that the $\SDE$ obtained in {\bf Step 1} is $\sum_{i = 0}^r P_i(x) f^{(i)}(x) = 0$.
For each value $e: \, (r+1)^2/2 \leq e \leq d + (r^2/2)$, we input in the $\SDE$ the polynomial $(x-Y)^e$, where $Y$ is a new variable; we obtain an equation of the form $g(x,Y) = 0$. We first observe that $g(x,Y) = (x-Y)^{e-r} h(x,Y)$, where $h(x,Y) \in \ZZ[x,Y]$ has degree $\leq r$. We write $h = \sum_{i = 0}^r h_i x^i$, where $h_i \in \ZZ[Y]$ is of degree $\leq r-i$.

The 
bit size of any coefficient of $h_i$ is $\overline{\mathcal O}(r^2\, size(M))$, which is $\overline{\mathcal O}(d\, size(M))$. Moreover,  since every $h_i \in \ZZ[Y]$ has degree $\leq r$, by \cite[Proposition 21.22]{Bostanetal}, the cost of computing the integer roots of each $h_i$ is $\overline{\mathcal O}(d^{\,2}\, size(M))$. 
Since we have to solve $r+1$ equations and take the common roots, this makes $\overline{\mathcal O}(d^{\,2.5} size(M))$ bitsize operations and since we have to do it for at most $d$ values of $e$, this gives $\overline{\mathcal O}(d^{\,3.5} size(M))$ bit operations overall.
Moreover, the $b_j$'s computed divide the independent term of all the $h_i$ and, hence, the 
bit size of each $b_i$ is $\overline{\mathcal O}(d\, size(M)).$

In {\bf Step 3}, the corresponding matrix has at most $d+1+(r^2/2)$ rows, at most $r$ columns (see the proof of Theorem \ref{th-algbigexp}), and its size is $\overline{\mathcal O}(d^{\,3.5} size(M))$. Since the rank of this matrix is $\leq r$ (indeed, as we proved in Theorem \ref{th-algbigexp}, this matrix has full column rank), when  we are performing Gaussian 
elimination and treating a new row we have at most $r$  already treated nonzero rows. As a consequence of this, we have to perform $\mathcal O(r^2 d)$ arithmetic operations to solve the system of equations by Gaussian elimination through pivoting. Hence the cost of this step is $\overline{\mathcal O}(d^{\,5.5} size(M))$, giving a total cost of $\mathcal{\mathcal O}(d^{\,6.5} size(f) + d^{\,8})$ bitsize operations.
\end{demo}

\bigskip

Now, we can proceed with the main result of this section:

\begin{th-algorithm}[Different nodes] \label{diffnodesimproved}
	Let $f \in \FF[x]$ be a polynomial that can be written as 
	\[
		f = \sum_{i = 1}^s \alpha_i (x - a_i)^{e_i},
	\]
	where the constants $a_i \in \FF$ are all distinct, $\alpha_i \in \FF \setminus \{0\}$, and  $e_i \in \NN$. 
Assume moreover that $n_i \leq (3i/4)^{1/3} - 1$ for all 
$i \geq 2$, where $n_i$ denotes the number of indices $j$ such that
 $e_j \leq i$.

	Then, $\AffPow_\FF(f) = s$. Moreover, there is a polynomial time algorithm {\tt Build}$(f)$ that receives $f = \sum_{i = 0}^d f_i x^i \in \FF[x]$ as input and computes the $s$-tuples of coefficients $C(f) = (\alpha_1,\ldots,\alpha_s)$, of nodes $N(f) = (a_1,\ldots,a_s)$ and exponents $E(f) = (e_1,\ldots,e_s)$.
	The algorithm {\tt Build}$(f)$ works as follows:

	\begin{enumerate}[\; \bf Step 1.]
		\item We take $t$ the minimum value such that $f$ satisfies a SDE$(t, 0)$ and compute explicitly one of these SDE.
	
		\item Consider $B := \{(x-b_i)^{d_i}\, \vert \, 1 \leq i \leq l\}$, the set of all the solutions of the SDE of the form $(x-b)^e$
		with $ (t+1)^2/2 \leq e \leq {\rm deg}(f) + \frac{({\rm deg}(f)+2)^2}{8}$ and assume that $d_1 \geq d_2 \geq \cdots \geq d_l \geq d_{l+1} := (t+1)^2/2$.
		
		\item We take $r \in \{1,\ldots,l\}$ such that $d_{r} - d_{r+1} > r^2/2$ and $d_{r+1} < {\rm deg}(f)$.
		
		\item We set $j := d_{r} - (r^2/2)$ and express $f^{(j)}$ as $f^{(j)} = \sum_{i = 1}^r \beta_i \frac{d_i !}{(d_i - j)!} (x-b_i)^{d_i - j}$
		with $\beta_1,\ldots,\beta_r \in \mathbb F$. We set $I := \{i \, \vert \, \beta_i \neq 0\}$.		
					 
		\item We set $\widetilde{f} := \sum_{i = 1}^{r} \beta_i (x-b_i)^{d_i}$ and $h := f - \widetilde{f}$. 
		
			If $h = 0$, then $C(f) = (\beta_i \, \vert  \, i \in I), \, N(f) = (b_i \, \vert  \, i \in I)$ and $E(f) = (d_i \, \vert  \, i \in I)$. 

			Otherwise, we set $h := f - \widetilde{f}$ and we have that $C(f) = (\beta_i \, \vert  \, i \in I)\, \cup \, C(h), \, N(f) = (b_i \, \vert  \, i \in I)\, \cup \, N(h)$
			and $E(f) = (d_i \, \vert \, i \in I)\, \cup \, E(h)$, where the triplet $(C(h),N(h),E(h))$ is the output of {\tt Build}$(h)$.
	\end{enumerate}
\end{th-algorithm}
\begin{proof}
	By Corollary \ref{uniquenessExpression} we have that $\AffPow_{\FF}(f) = s$. Concerning the algorithm, first we observe that the value $t$ computed in {\bf Step 1} 
	is $\leq 2s-1$ by Proposition \ref{smallequation}. Moreover, we claim that the
	set $B$ computed in {\bf Step 2} has $l \leq t$ elements. Otherwise, by Lemma \ref{dimsolution}, 
	there exists a set $I \subseteq \{1,\ldots,l\}$ of size $\leq t+1$ 
	and there exist
		$\{\gamma_i \, \vert \, i \in I\} \subseteq \FF \setminus \{0\}$ 
	such that
	$\sum_{i \in I} \gamma_i (x-b_i)^{d_i} = 0$. Setting $m := {\rm max}\{d_i\, \vert \, i \in I\} \geq (t+1)^2/2$,
	Proposition \ref{LowerBound} yields that $t+1 \geq \vert I \vert > \sqrt{2(m+1)} > t+1$, a contradiction.

	Now we set $L := 5s^2/ 2$ and consider the set $C := \{(x-a_i)^{e_i} \, \vert \, e_i \geq L\}$  where the $a_i$'s are the nodes in the optimal expression of $f$.
	We have that $C \not= \emptyset$; indeed, if we set
	 $e_{\max} := \max\{e_i \,\vert \, 1 \leq i \leq s\}$, then $s = n_{e_{\max}} \leq (3e_{\max}/4)^{1/3}-1$ and $L  \leq 4(s+1)^3/3  \leq e_{\max}.$ 
	 
	Since 
	$$ t s + \binom{s}{2} \leq (2s-1)s + \binom{s}{2} \leq 5s^2/2,$$
	Proposition \ref{bigexponents} yields that all the elements of $C$ are solution of the SDE and, by Corollary \ref{upperbounddegreeterms} we
	know that $e_i \leq {\rm deg}(f) + \frac{({\rm deg}(f)+2)^2}{8}$ for all $i \in \{1,\ldots,s\}$. Hence $C \subseteq B$. In particular, there exists
	a $\tau \in \{1,\ldots,l\}$ such that $d_1 \geq d_{\tau} = e_{\max} \geq \frac{4}{3} (s+1)^3$.

	Now we take $k := {\max}\{i\, \vert \, d_i > L\}$ (we have that $1 \leq k \leq l \leq t \leq 2s-1$) and we are going to prove that 
	\begin{itemize}
		\item there exists $r \in \{\tau,\ldots,k-1\}$ such that $d_r - d_{r+1} > r^2/2$, or
		\item $d_k - L >  k^2/2$.
	\end{itemize} 
	Indeed, if this is not the case, then we get the following contradiction:
		$$
			\begin{array}{rl} \frac{4s^3}{3} \leq  & \frac{4(s+1)^3}{3} - L \leq e_{\max} - L = d_{\tau} - L = \sum_{i = \tau}^{k-1} (d_i - d_{i+1}) + d_k -  L \leq \\
			\leq & \frac{1}{2}\sum_{i=\tau}^k i^2 \leq  \frac{1}{2}\sum_{i=1}^k i^2  = \frac{k(k+1)(2k+1)}{12} \leq  \frac{(2s-1) 2s (4s-1)}{12} < \frac{4s^3}{3}. \end{array} 
		$$
	
	We take $r \in \{1,\ldots,k-1\}$ such that $d_r - d_{r+1} > r^2/2$, or $r = k$ if such a value does not exist (and $d_k - L > k^2/2$).  We claim that $e_{\max} \geq d_r$
	if and only if  $d_{r+1} < {\rm deg}(f)$ and, thus, the $r$ described in {\bf Step 3} always exists. If $d_{r+1} < {\rm deg}(f)$, since ${\rm deg}(f) \leq e$
	and $C \subseteq B$, then $d_r \leq e_{\max}$ 
	(since $e_{\max} = d_{\tau}$, it cannot be sandwiched between two consecutive elements $d_r,d_{r+1}$ of this sequence).

	Conversely,
	assume now that $e_{\max} \geq d_r$ and let us prove that $d_{r+1} < {\rm deg}(f)$.  To prove this we first observe that
	setting $j := d_{r}-(r^2/2)$, then $f^{(j)}$ can be uniquely expressed as a linear combination of $B' := \{(x-b_i)^{d_i - j} \,\vert \,\, 1 \leq j \leq r\}$ .
	Indeed, $f^{(j)} = \sum_{e_i \geq j} \alpha_i \frac{e_i !}{(e_i - j)!} (x-a_i)^{e_i-j}$ with
	$\alpha_i \not= 0$ and $(x-a_i)^{e_i-j} \in B'$ for all $e_i \geq j$, and if there is another way of expressing $f^{(j)}$ as a linear 
	combination of $B'$, then by Proposition \ref{LowerBound} we get that $r > \sqrt{2 ({\rm min}\{d_i\, 
	\vert \, 1 \leq i \leq r\} - j + 1)} \geq \sqrt{r^2 + 2} > r$, a contradiction. So, if $d_{r+1} \geq {\rm deg}(f)$,
	then $f^{(j)} = 0$ and the only expression of $f^{(j)}$ as a linear
	combination of $B'$ would be the one in which every coefficient is $0$, a contradiction.
	Hence, the value $r$ computed in {\bf Step 3} exists.
	
	We have seen that 
	$f^{(j)}$ can be uniquely expressed as a linear combination of $B'$ as $f^{(j)} = \sum_{e_i \geq j} \alpha_i \frac{e_i !}{(e_i - j)!} (x-a_i)^{e_i-j}$.
	Hence, in {\bf Step 4}, one finds all the $(\alpha_i, a_i, e_i)$ such that $e_i \geq j$. In {\bf Step 5}, either the polynomial $h$ is $0$ and we have finished
	or $h = \sum_{e_i < j} \alpha_i (x-a_i)^{e_i}$ is written as a linear combination  of strictly less than $s$ terms  and satisfies the hypotheses of the Theorem, so by induction we are done. 
\end{proof}

{  Note that in Step~2 of this algorithm we need to compute polynomial roots,
just as in the corresponding step of Theorem~\ref{th-algbigexp} (see the proof 
of Theorem~\ref{th-algbigexp} for details). One difference, however, 
is that we do not use the roots $b_i$ only to output the coefficients 
of the optimal decomposition: we also use the $b_i$ in the subsequent iterations of the algorithm since the polynomials $\tilde{f}$ and $h$ of Step~5 are defined in terms of the $b_i$, and we call the algorithm recursively on input~$h$.
From this discussion one might be lead to think that if $f$ 
has its coefficients in a subfield $\KK$ of $\FF$, the coefficients of 
$\tilde{f}$ and $h$ may lie outside $\KK$. We show in Proposition~\ref{Galois} that this is not the case: $\tilde{f}$ 
and $h=f-\tilde{f}$ always lie in $\KK[X]$. We do not know if $\tilde{f}$ can be computed from $f$ with a polynomial number of arithmetic operations and 
comparisons (in the words of Section~\ref{modcomp}, this would be a way 
to eliminate root finding from the ``internal working'' of the algorithm).}

\begin{proposition} \label{Galois}
	Let $\KK$ be a subfield of $\FF$. Let $f \in \KK[x]$ be a polynomial that can be expressed in the $\AffPow_\FF$ model as
		\[ f = \sum_{i = 1}^s \alpha_i (x-a_i)^{e_i} {\text \ with \ } \alpha_i, a_i \in \FF, \]
	and $n_e = \#\{i : e_i \leq e\} \leq \sqrt{\frac{e+1}{2}}$ for all $e \in \NN$.
	Then, for all $m, M \in \NN$, the truncated expression 
		\[ \widetilde{f} = \sum_{m \leq e_i \leq M} \alpha_i (x-a_i)^{e_i} \]
	belongs to $\KK[x]$. In particular, whenever  $f \in \FF[x]$ satisfies the hypotheses
	of Theorem \ref{diffnodesimproved} and $f \in \KK[x]$, then the polynomial $\widetilde{f}$ computed in {\bf Step 5} of the algorithm also belongs to $\KK[x]$.
\end{proposition}
\begin{proof} 
	By Corollary \ref{uniquenessExpression}, we know that $\AffPow_\FF(f) = s$ and, hence, $\alpha_i, a_i$ are algebraic over $\KK$.
	We denote by $\TT$ the splitting field of the minimal polynomials of all the
	$\alpha_i, a_i$ over $\KK$ (i.e., the smallest field $\TT$ such that $\KK(\alpha_i,a_i) \subset \TT$ and $\KK \subset \TT$ is normal).
	Since $\KK$ is of characteristic $0$ (and, thus, the
	extension $\KK \subset \TT$ is separable), then $\KK \subset \TT$ is a Galois extension. 
	
	Take now $\sigma$ any element of the Galois group of the extension $\KK \subset \TT$.
	Since $f \in \KK[x]$, if we apply $\sigma$ to $f$ we obtain that $f = \sigma(f)  = \sum_{i = 1}^s \sigma(\alpha_i) (x- \sigma(a_i))^{e_i}$.
	Moreover, by Corollary \ref{uniquenessExpression},
	we know that $\AffPow_\TT(f) = s$ and $f$ has a unique optimal expression in the $\AffPow_\TT$ model, then $\{(\alpha_i, a_i, e_i) \, \vert \, 1 \leq i \leq s\} =
		\{(\sigma(\alpha_i), \sigma(a_i), e_i) \, \vert \, 1 \leq i \leq s\}. $
	 In particular, for every $e \in \NN$, we have that 
	 	\begin{equation}\label{permut} \{(\alpha_i, a_i, e_i) \, \vert \, e_i = e\} =
		\{(\sigma(\alpha_i), \sigma(a_i), e_i) \, \vert\, e_i = e\}.\end{equation}
	Now, we consider $\widetilde{f} = \sum_{m \leq e_i \leq M} \alpha_i (x-a_i)^{e_i}$,  by (\ref{permut}) we get that 
		\[ \sigma(\widetilde{f}) = \sum_{m \leq e_i \leq M} \sigma(\alpha_i) (x-\sigma(a_i))^{e_i} = 
		\sum_{m \leq e_i \leq M} \alpha_i (x-a_i)^{e_i} = \widetilde{f}.\]
	Summarizing, if we denote $\widetilde{f} = \sum_{i = m}^M f_i x^i \in \TT[x]$,
	we have proved that $\sigma(f_i) = f_i$ for every $i \in \{0,\ldots,M\}$ and every $\sigma$ in the Galois group of the extension $\KK \subset \TT$. This 
	 proves (see, e.g., \cite[Theorem 7.1.1]{Cox}) that $f_i \in \KK$ for all $i \in \{0,\ldots,M\}$ and $\widetilde{f} \in \KK[x]$.

\end{proof}

{ 

We define the size of the set of triplets $\{(\alpha_i,a_i,e_i) \, \vert \, 1 \leq i \leq s\} \subset \ZZ \times \ZZ \times \NN$
as $\sum_{i = 1}^s [1 + \log_2(1 + |a_i|) + \log_2(1 + |\alpha_i|) + e_i]$. As mentioned in the introduction,
it is not clear that the size of the output of the algorithm proposed in Theorem \ref{diffnodesimproved} is polynomially bounded 
in the input size (i.e., in the bit size of $f$ given as a sum of monomials). However, it is straightforward to check that
the input size is polynomially bounded by the output size. Indeed, the degree of $f$ is upper bounded by the maximum value of the $e_i$ and 
every coefficient of $f$ can be seen as the evaluation
of a small polynomial in the $\alpha_i, a_i$'s. In the following result we prove that the algorithm works in polynomial time in 
the size of the output. Hence, a positive answer to Question \ref{questinputoutput} together with Corollary \ref{upperbounddegreeterms} would directly
yield that the algorithm works in polynomial time (in the size of the input).

\begin{proposition}\label{polyoutput}
Let $f \in \ZZ[x]$ be written as 
	$$f = \sum_{i = 1}^s \alpha_i (x - a_i)^{e_i}$$ 
with $a_i \in \ZZ$, $\alpha_i \in \ZZ \setminus \{0\}$, $e_i \in \NN$
and {  assume} that this decomposition satisfies the conditions of Theorem~\ref{diffnodesimproved}:  the constants $a_i$ are all distinct, and  $n_i \leq (3i/4)^{1/3} - 1$ for all 
$i \geq 2$, where $n_i$ denotes the number of indices $j$ such that $e_j \leq i$. 

Then, the algorithm in Theorem \ref{diffnodesimproved} works
in polynomial time in the size of the output.\end{proposition}

\begin{proof} 
We write $f = \sum_{j = 0}^d f_j x^j$ with $f_j \in \ZZ$ and $d = {\rm deg}(f) \leq {\rm max}\{e_1,\ldots,e_s\}$. 
We have that $f_j = \sum_{e_i \geq j} \alpha_i \binom{e_i}{j} a_i^{e_i-j}$ for all $j \in \{0,\ldots,d\}$. Thus, the size of $f$ is polynomially bounded by the size of the output.
To perform {\bf Step 1} we follow Remark \ref{findSDE}. We note that the coefficients of the polynomials appearing in the $\SDE$ are polynomially bounded by
the size of $f$. In {\bf Step 2}
we have to compute the integral roots of polynomials of degree $t \leq s$ with integral coefficients, which can also be done in polynomial time (see, e.g., \cite{LLL}).
{\bf Step 4} can also be performed in polynomial time by solving a linear system of equations (see, e.g., \cite[Corollary 3.3a]{Schrijver}) . The result follows from the fact that 
the polynomial $h$ defined
in {\bf Step 5} can be written as $h = \sum_{j \in J} \alpha_j (x-a_j)^{e_j}$ for some set $J \subset \{1,\ldots,s\}$ of at most $s-1$ elements.
{  After the first iteration, the algorithm is therefore 
 called recursively on polynomials $h$ with an output size bounded by the
output size of the original $f$.}
\end{proof} }

\section{Algorithms for repeated nodes} \label{repeatsec}

 This section is a continuation of the previous one and concerns the case where the nodes $a_i$ in the optimal expression of $f$ in the Affine Power model are not necessarily 
 different. The section is divided in two. In the first subsection we provide algorithms when 
all the exponents corresponding to a repeated node appear in a small interval.
The second one handles the case where the difference between 
two consecutive exponents corresponding to the same node is always large.

\subsection{Small intervals} \label{verybig}

We begin with the following result generalizing Proposition \ref{bigexponents}, which corresponds to $\delta = 0$.
\begin{proposition}\label{bigExpRepeatedNodesRootSDE}
	Let $\delta \in \NN^+$ and let $f \in \FF[x]$ be written as
	\[
		f = \sum_{i = 1}^t Q_i(x)\, (x-a_i)^{e_i},
	\]
	with distinct $a_i \in \FF$, $Q_i \in \FF[X]$ with $\deg(Q_i) \leq \delta$ and $e_i \in \NN$ for all $i$.
	Assume that $f$ satisfies the following SDE of parameters $k,l$: 
	\[
		\sum_{i = 0}^k P_i(x) f^{(i)}(x) = 0.
	\]
	If $e_i \geq  k + (t - 1)(k + l + \delta) + \binom{t}{2}$, then $Q_i(x)\,(x-a_i)^{e_i}$ satisfies the same SDE; as a consequence $P_k(a_i) = 0$.
\end{proposition}
\begin{proof}
	We take $i = 1$. We assume that $e_1 \geq  k + (t - 1)(k + l + \delta) + \binom{t}{2}$
	and that $f$ satisfies a $\SDE(k,l)$
	\[
		\sum_{i = 0}^k P_i(x) f^{(i)}(x) = 0
	\]
	By contradiction, suppose that $Q_1(x)(x-a_1)^{e_1}$ does not satisfy this equation.
	For every $j \in \{1, \ldots, t\}$, we denote by $g_j$ and $R_j$ the polynomials such that
	\[
		g_j = \sum_{i = 0}^k P_i(x)\, (Q_j(x)(x - a_j)^{e_j})^{(i)} = R_j(x)(x - a_j)^{d_j}
	\]
	where $d_j := \max\{0,e_j - k\}$ for all $j$, and with $\deg R_j \leq  \delta + e_j + l - d_j \leq k + l + \delta$.
	We have the equality
	\begin{equation} \label{eq1}
		-g_1 = \sum_{i = 2}^t g_i \neq 0.
	\end{equation}
	We consider a linearly independent subfamily of $g_2,\ldots,g_t$, namely $\{g_j \, \vert \, j \in J\}$ with
	$J = \{j_1,\ldots,j_p\} \subseteq \{2,\ldots,t\}$. 
	Then by Proposition \ref{factorWronskian} we have that
	\[
		\begin{array}{cl} e_1-k = d_1  \leq \mult[a_1]{g_1} & \leq p - 1 + \sum_{j \in J}{\rm deg}(R_j) + p(p-1) - \binom{p}{2} \\ & \leq p-1 + (k+l + \delta)p + \binom{p}{2}. \end{array}
	\] 
	Taking into account that $p \leq t-1$, we finally obtain the inequality
	\[
		e_1 - k \leq t - 2 + (t - 1)(k + l + \delta) + \tbinom{t - 1}{2},
	\]
	which yields a contradiction.
	
	Now, we take $l_1 \geq e_1$ and $R_1 \in \FF[x]$ such that $(x-a_1)^{l_i}\, R_1(x) = (x-a_1)^{e_1}\, Q_1(x)$ and $R_1(a_1) \neq 0$.
	Since $(x-a_1)^{e_1}\, Q_1(x)$ is a solution of the $\SDE$, we have that:
 	\[
 		\sum_{i = 0}^k P_i(x)\, ((x-a_1)^{l_1}\, R_1(x))^{(i)} = 0,
 	\]
 	we deduce that there exists $q \in \FF[x]$ such that $P_k(x) (x-a_1)^{l_1-k} h(x) = (x-a_1)^{l_1-k+1} q(x)$, from where we deduce that $P_k(a_1) = 0$.
\end{proof}

\bigskip

From Proposition \ref{bigExpRepeatedNodesRootSDE} we shall now derive Corollary \ref{corBigExpRepeatedNodesRootSDE} and Theorem \ref{th-algVeryLargeExp}. They provide an effective
method to obtain the optimal expression of a polynomial $f$ in the Affine Power model whenever 
all the exponents corresponding to a repeated node are required to 
lie in a small interval.

\begin{corollary}\label{corBigExpRepeatedNodesRootSDE}
	Let $\delta \in \ZZ^+$ and let $f \in \FF[x]$ be a polynomial written as
	\[
		f = \sum_{i = 1}^t Q_i(x)\, (x-a_i)^{e_i},
	\]
	where:
	\begin{itemize} 
		\item $Q_i(x) = \sum_{j = 1}^{s_i} \gamma_{i,j} (x-a_i)^{\epsilon_{i,j}} \in \FF[x]$ with $\gamma_{i,j} \neq 0$ and 
		$0 = \epsilon_{i,0} < \epsilon_{i,1} < \cdots < \epsilon_{i,s_i} \leq \delta$,
		\item the $a_i$'s are elements of $\FF$ and are all distinct, and
		\item $e_i \geq 5 t^2 (\delta+1)^2 / 2$ for all $i$.
	\end{itemize}
	Then, 
	\begin{enumerate}[\quad a)]
		\item the set of polynomials $\{Q_i(x)\, (x-a_i)^{e_i} \, \vert \, 1 \leq i \leq t\}$ is linearly independent,

		\item $\AffPow_\FF(f) = \sum_{i = 1}^t s_i$ and the optimal representation of $f$ is unique,

		\item $f$ satisfies a $\SDE(2t-1,\delta)$,
		
		\item if $f$ satisfies the $\SDE(k,\delta)$
			$$\sum_{i = 0}^k P_i(x) f^{(i)}(x) = 0$$
		and $k \leq 2t-1$; then $Q_i(x) (x-a_i)^{e_i}$ also satisfies it and $P_r(a_i) = 0$ for all $i \in \{1,\ldots,s\}$, and
		
		\item $f$ does not satisfy any $\SDE(k,\delta)$ with $k < t$.
	\end{enumerate}
\end{corollary}
\begin{proof}
	Notice that (b) implies (a). To prove (b), we observe that $f$ is written as
		$$f = \sum_{i = 1}^t \sum_{j = 1}^{s_i} \gamma_{i,j}\, (x-a_i)^{e_i + \epsilon_{i,j}},$$
	so $\AffPow_\FF(f) \leq \sum_{i = 1}^t s_i$. Now assume that $f$ can also be expressed as
		$f = \sum_{i = 1}^r \beta_i (x-b_i)^{d_i}$ with $\beta_i \in \FF$ and $r \leq \sum_{i = 1}^t s_i \leq t (\delta + 1)$. By Proposition \ref{LowerBound}
	we get that either both expressions are the same, or 
		$$2t(\delta + 1) \geq r + \sum_{i = 1}^t s_i \geq \sqrt{2 ({\rm min}\{e_1,\ldots,e_t\} + 1)} >\sqrt{ 5 t^2 (\delta + 1)^2},$$
	which is not possible. Thus $\AffPow_\FF(f) = \sum_{i = 1}^t s_i$ and the optimal representation of $f$ is unique.

	From Proposition \ref{smallequation} we get (c). If $f$  satisfies a $\SDE(k,d)$ with $k \leq 2t-1$,
	then for all $i \in \{1,\ldots,t\}$ we have that
	\begin{equation} \label{ineqBigRepeated}
		\begin{array}{ll} e_i &\geq \frac{5}{2} t^2 (\delta+1)^2 > \frac{5}{2}t^2 - \frac{3}{2}t + 2t \delta - 2\delta  \\
			&= 2t-1 + (t - 1)\left(2t-1+2\delta \right) + \binom{t}{2} \\ & \geq k + (t - 1)\left(k+2\delta \right) + \binom{t}{2}. \end{array}
	\end{equation} 
	Hence, Proposition \ref{bigExpRepeatedNodesRootSDE} yields that $Q_i(x) (x-a_i)^{e_i}$ is also a solution of this
	equation for all $i$, proving (d). Finally, $f$ cannot satisfy a $\SDE(k,\delta)$ with $k < t$; otherwise by (a) and (d), the vector space of solutions to this equation 
	has  dimension $\geq t$, which contradicts Lemma \ref{dimsolution}.

	\end{proof}

\begin{th-algorithm}[repeated nodes in small intervals]\label{th-algVeryLargeExp}

	Let $\delta \in \ZZ^+$ and let $f \in \FF[x]$ be a polynomial of degree $d$ that can be written as
	\[
		f = \sum_{i = 1}^t Q_i(x)\, (x-a_i)^{e_i},
	\]
	with
	\begin{itemize} 
		\item $Q_i(x) = \sum_{j = 1}^{s_i} \gamma_{i,j} (x-a_i)^{\epsilon_{i,j}} \in \FF[x]$ with $\gamma_{i,j} \neq 0$ and 
		$0 = \epsilon_{i,0} < \epsilon_{i,1} < \cdots < \epsilon_{i,s_i} \leq \delta$,
		\item the $a_i$'s are elements of $\FF$ and are all distinct, and
		\item $e_i \geq  \frac{5}{2}t^2(\delta+1)^2$ for all $i$.
	\end{itemize}
	Then $\AffPow_\FF(f) = \sum_{i = 1}^t s_i$. Moreover, there is a polynomial time algorithm  {\tt Build}$(f,\delta)$ that
	receives $f = \sum_{i = 0}^d f_i x^i \in \FF[x]$ and $\delta$ as input
	and computes the $t$-tuples of nodes $N(f) = (a_1,\ldots, a_t)$, the values $s_1,\ldots,s_t$ and the tuple of coefficients
	$C(f) = ( \gamma_{i,j} \, : \, 1 \leq i \leq t, \, 1 \leq j \leq s_i)$,  and exponents $E(f) = (e_i + \epsilon_{i,j}\, : \, 1 \leq i \leq t, \, 1 \leq j \leq s_i)$. 
	The algorithm {\tt Build}$(f,\delta)$ works as follows:

	\begin{enumerate}[\; \bf Step 1.]
		\item Take $r$ the minimum value such that $f$ satisfies a $\SDE(r, \delta)$. Compute explicitly one of these SDE, i.e., compute $P_0,\ldots,P_r \in \FF[x] $ such that
		$\sum_{i = 0}^r P_i(x) f^{(i)}(x) = 0$ and $\deg(P_i) \leq i + \delta$.

		\item Compute $\mathcal R = \{c_1,\ldots,c_p\} \subseteq \FF$ the set of roots of $P_r$. 
		
		For each $i \in \{1,\ldots,p\}$, consider the $\FF$-vector space 
		$V_{i}$ spanned by the solutions of the $\SDE$ of the form $R(x) (x - c_i)^e$,
		with $\frac{(r+1)^2(\delta+1)^2 }{2} < e < d + \frac{r^2(\delta+1)^2}{2}$ and $R(x)$ a polynomial of degree $\leq \delta$. 
		
		We take $B_i = \{g_{i,1},\ldots,g_{i,l_i}\}$ a base of $V_i$, where $g_{i,j} = R(x) (x-c_i)^e$ with $\frac{(r+1)^2(\delta+1)^2 }{2} < e <
		d + \frac{r^2(\delta+1)^2}{2}$ and ${\rm deg}(R(x)) \leq \delta$. We set $B := \cup_{i = 1}^p B_{i}$.
		
		\item Express $f$ as a linear combination of the elements of $B$, namely, $f = \sum_{i = 1}^p \sum_{j = 1}^{l_i}
		\lambda_{i,j}\, g_{i,j}$ with $\lambda_{i,j} \in \FF$.
		
		\item Denote $f_i = \sum_{j = 1}^{l_i} \lambda_{i,j}\, g_{i,j}$, for all $i \in \{1,\ldots,p\}$. Write $f_i$ in the shift $c_i$, i.e.,
		$f_i = \sum_{j = 1}^{r_i} \beta_{i,j}(x - c_i)^{\mu_{i,j}}$
		with $\beta_{i,j} \in \FF \setminus \{0\}$.
		
		\item Output $N(f) = (c_1,\ldots,c_p)$, $r_1,\ldots,r_p \in \NN$, $C(f) = (\beta_{i,j} \,\vert \, 1 \leq i \leq p, 1 \leq j \leq r_i)$ and
		$E(f) = (\mu_{i,j} \,\vert \, 1 \leq i \leq p, 1 \leq j \leq r_i).$
	\end{enumerate}
\end{th-algorithm}
\begin{proof}
	We observe that $f$ satisfies the hypotheses of Corollary \ref{corBigExpRepeatedNodesRootSDE}; then, by Corollary \ref{corBigExpRepeatedNodesRootSDE}.(b),
	we have that $\AffPow_\FF(f) = \sum_{i = 1}^t s_i$ and that there is a unique optimal expression of $f$ in the $\AffPow$ model. 
	
	Let us prove the correctness of the algorithm {\tt Build}$(f,\delta)$. By Corollary \ref{corBigExpRepeatedNodesRootSDE}.(c), (d) and (e), the value $r$ computed in
	{\bf Step 1} satisfies that $t \leq r \leq 2t-1$.
	For all $i \in \{1,\ldots,t\}$ we have that
	\begin{itemize} 
		\item $a_i \in \mathcal R$, and 
		\item $Q_i(x)\, (x-a_i)^{e_i}$ is a solution of the $\SDE$ computed in {\bf Step 1} 
	\end{itemize} 
	
	Moreover, the input polynomial $f$ can be expressed as a linear combination of the elements of $B$, because:
	\begin{itemize}
		\item $f$ can be written as a combination of $Q_i(x)\, (x-a_i)^{e_i}$.
		\item Since $\AffPow_\FF(f) = \sum_{i = 1}^t s_i$ and $s_i \leq \delta$, we have $\AffPow_\FF(f) \leq t(\delta + 1) \leq r(\delta + 1)$ and thus using Corollary \ref{upperbounddegreeterms} we have $\max\{e_i\} < d + \frac{r^2(\delta + 1)^2}{2}$.
		On the other hand, we have $e_i \geq \frac{5}{2} t^2 (\delta+1)^2 > 2 t^2 (\delta+1)^2 \geq \frac{(r+1)^2 (\delta+1)^2}{2}$.
		This implies that $Q_i(x)\, (x-a_i)^{e_i}$ belongs to $V_i$ and thus can be written as a linear combination of the elements of $B_i$.
	\end{itemize}
	
	So, let us assume (we will prove it later) that all the elements of $B$ are linearly independent. Then, in {\bf Step 3} there is a unique way of writing of
	$f$ as a linear combination of the elements of $B$. Finally, it suffices to  write $f_i = R_i(x) (x-c_i)^{d_i}$ and consider the Taylor expansion of $R_i(x)$ with respect to $c_i$
	for every $i \in \{1,\ldots,p\}$ as in {\bf Step 4} to get the desired sets of nodes, coefficients and exponents.

	To prove the correctness of the algorithm, it only remains to prove that the elements of $B$ are linearly independent. To prove this we will follow a similar argument
	to that of Propposition \ref{LowerBound}. Assume that the elements of $B$ are not linearly independent. We take $W = \{P_{i}(x) (x-b_i)^{d_i} \, \vert \, 1 \leq i \leq w\} \subset B$ a minimal $\FF$-linearly dependent set. 
	By Lemma \ref{dimsolution}, the size of this set is $w \leq r+1 \leq 2t$. Then, there exist $\lambda_1,\ldots,\lambda_w \in \FF \setminus \{0\}$
	such that $\sum_{i=1}^w \lambda_i P_{i}(x) (x-b_i)^{d_i} = 0$. We set $Z := \{i \, \vert \, b_i \neq b_1\}$.
	We also set $\tau := {\rm min}\{d_i \, \vert \, b_i = b_1\}$ and take $R(x)$ such that
		$$ R(x) (x-b_1)^{\tau} = \sum_{i \notin Z} - \lambda_i P_{i}(x) (x-b_1)^{d_i} = \sum_{i \in Z} \lambda_i P_{i}(x) (x-b_i)^{d_i}.$$
	We observe that $R(x) \neq 0$ because $\{P_{i}(x) (x-b_1)^{d_i} \, \vert \, i \notin Z\}$ is $\FF$-linearly independent.
	We assume that $Z = \{b_2,\ldots,b_{z+1}\}$ with $z \leq w-1 \leq r$. For all $j \in \{0,\ldots,z-1\}$, if we differentiate the above expression $j$ times we have that 
	\begin{equation}\label{cramerrule} R_j(x) (x-b_1)^{\tau-j} =  \sum_{i = 2}^{z+1} \lambda_i P_{i,j}(x) (x-b_i)^{d_i-j},
	\end{equation}
	where $P_{i,j}(x) \in \FF[x]$ are polynomials of degree $\leq \delta$. We set $g_i := P_i(x) (x-b_i)^{d_i}$ for all $i \in \{2,\ldots,z+1\}$ and 
	apply Cramer's rule to the system of equations (\ref{cramerrule}) to get that 
		$$ \lambda_1 = \frac{\Wr(R(x) (x-b_1)^{\tau}, g_3,\ldots,g_{z+1})}{\Wr(g_2,\ldots,g_{z+1})}$$.
		
	We observe that $$ \Wr(R(x) (x-b_1)^{\tau}, g_3,\ldots,g_z) =   (x-b_1)^{\tau - (z - 1)} \prod_{i=3}^{z+1} (x-b_i)^{d_i - (z-1)} W_1,$$ and  $$\Wr(g_2,\ldots,g_z) = 
	\prod_{i = 2}^{z+1} (x-b_i)^{d_i - (z-1)} W_2,$$
	where $W_2$ is a polynomial of degree $\leq z \delta + \frac{z (z-1)}{2}$.
	Thus, 
		$$(x-b_1)^{\tau - (z - 1)} W_1 = (x-b_2)^{e_2 - (z-1)} W_2$$ and, 
	 the multiplicity of $b_1$ in both sides of this expression has to be the same. Therefore, 
	 $\tau - (z-1) \leq \mult[b_1]{(x-b_1)^{\tau - (z - 1)}W_1} = \mult[b_1]{(x-b_2)^{e_2 - (z - 1)}W_2}  \leq {\rm deg}(W_2) \leq z \delta + \frac{z (z-1)}{2}$,
	 and $\tau \leq (z-1) + z \delta + \frac{z (z-1)}{2} \leq r - 1 + r \delta + \frac{r(r-1)}{2} \leq  \frac{(r+1)^2}{2} + r \delta \leq \frac{(r+1)^2  (\delta+1)^2}{2}$, a contradiction.

\end{proof}

\begin{remark} 
	The algorithm ${\tt Build}(f,\delta)$ described above can be slightly modified to not receive $\delta$ as input as long as $f$
	satisfies the hypotheses of Theorem \ref{th-algVeryLargeExp} for some $t , \delta \in \ZZ^+$. That is, we only need to assume that there exists $\delta \in \ZZ^+$ such that
		\[
		f = \sum_{i = 1}^t Q_i(x)\, (x-a_i)^{e_i},
		\]
	where $Q_i(x)\in \FF[x]$ has degree $\leq \delta$, the $a_i$'s are distinct elements of $\FF$, and
	$e_i \geq \frac{5}{2}t^2(\delta + 1)^2$ for all $i$. Indeed, it suffices to start with $\delta = 0$ and 
execute ${\tt Build}(f, \delta)$ with increasing values of $\delta$ 
until the reconstruction of $f$ succeeds.
The correctness
	of this algorithm is justified by Corollary \ref{corBigExpRepeatedNodesRootSDE}.(b). In fact, once we find~$\delta$ such that the reconstruction is possible, 
	we obtain the optimal expression of $f$ in the Affine Power model.
\end{remark}

\subsection{Big gaps} 

This subsection deals with polynomials $f$ such that whenever the terms $(x-a)^e$ and $(x-a)^d$ appear in the optimal expression of $f$ in the
Affine Power model, then the difference between $d$ and $e$ is ``large''. 
Similarly to Section~\ref{verybig}, we begin with some
results ensuring that whenever $f$ satisfies
a $\SDE$, then so do some of its terms in the optimal expression of $f$ in the Affine Power model. The desired algorithm then follows as a consequence 
of these results.

\begin{proposition}\label{repeatedNodesRootSDE}
	Let $f \in \FF[x]$ be written as
	\[
		f = (x - a)^{m} g(x) + \sum_{i = 1}^s \alpha_i (x-a)^{e_i} + \sum_{i = 1}^{p} \beta_i (x - a_i)^{d_i},
	\]
	with $g \in \FF[x], a, a_i, \alpha_i, \beta_i \in \FF$, $m, e_i, d_i \in \NN$ and $a_i \neq a$ for all $i$.
	We set $e := {\rm max}\{e_1,\ldots,e_s\}$ if $s \geq 1$ or $e := -1$ if $s = 0$.
	Whenever $f$ satisfies a $\SDE(k,l)$ with $m - e > p + (k+l)(p+1) + \binom{p}{2}$, then $(x-a)^m g$
	satisfies the same SDE.
\end{proposition}
\begin{proof}
	Assume that $f$ satisfies a $\SDE(k,l)$
	\[
		\sum_{i = 1}^k P_i(x) f^{(i)}(x) = 0
	\]	
	By contradiction, we assume that $(x - a)^m g(x)$ does not satisfy this equation. Thus, there exists $T(x) \in \FF[x]$ nonzero such 
	that 	$\sum_{i = 0}^k P_i(x)\, ((x-a)^m g)^{(i)} = T(x) (x - a)^{m-k}$.
	For every $j \in \{1, \ldots, s\}$ and every $j \in \{1,\ldots,p\}$, we denote by $h_j$ and $g_j$ the polynomials such that
	\[
		h_j = \sum_{i = 0}^k P_i(x)\, ((x-a)^{e_j})^{(i)}
		\quad \text{ and } \quad 
		g_j = \sum_{i = 0}^k P_i(x)\, ((x-a_j)^{d_j})^{(i)}.
	\]
	We observe that $\deg(h_j) \leq e_j + l \leq e + l$ and $\deg(g_j) \leq d_j + l$. Since $f$ satisfies the already mentioned $\SDE$, we get that 
	\[
T(x) (x - a)^{m-k} = \sum_{i = 1}^s \alpha_i h_i + \sum_{i = 1}^p \beta_i g_i.	\]
	If we differentiate  $(e + l + 1)$ times on both sides of the previous equation, we obtain an equality of the following form
	\[
U(x) (x - a)^{m-k-e-l-1} = \sum_{i = 1}^p \beta_i g_i^{(e + l + 1)} 			= \sum_{i = 1}^p U_i(x) (x - a_i)^{r_i}
	\]
	with $r_i := \max\{0, d_i-k-e-l-1\}$ and $\deg(U_i(x)) \leq k + l$.
	If we take a linearly independent family $\{g_i^{(e + l + 1)} \, : \, i \in I\} \subseteq \{g_i^{(e + l + 1)} \, : \, i \in \{1,\ldots,p\}\}$ 
	and compute the multiplicity of $a$ on both sides of the previous equality using Proposition \ref{factorWronskian}, we obtain that
	\[
		m - k - e - l - 1 \leq  p - 1 + (k+l)p + (p-1)p - \binom{p}{2},
	\]
	which yields that
	\[
		m  - e \leq  p + (k+l)(p + 1)  + \binom{p}{2},
	\]
	a contradiction.
\end{proof}

\medskip

The following result is a generalization of Proposition \ref{bigexponents} where we allow repeated nodes provided their corresponding exponents are far enough.

\begin{corollary} \label{biggapsolution}
	Let $f \in \FF[x]$ be written as
	\[
		f = \sum_{i = 1}^s \alpha_i (x-a)^{e_i} + \sum_{i = 1}^{p} \beta_i (x - a_i)^{d_i},
	\]
	with $a, a_i, \alpha_i, \beta_i \in \FF$, $m, e_i, d_i \in \NN$, $a_i \neq a$ for all $i$ and $e_s > \dots > e_1 > e_0 := -1$.
	Assume that $f$ satisfies a $\SDE(k,l)$ and that $e_{i+1} - e_i > p + (k+l)(p+1) + \binom{p}{2}$ for all $i$,
	then $(x-a)^{e_i}$ satisfies the same $\SDE$ for all $i \in \{1,\ldots,s\}$.
\end{corollary}
\begin{proof}
	Assume that there exists an $e_i$ such that $(x-a)^{e_i}$ does not satisfy the $\SDE(k,l)$ and we take $e$ the maximum of such $e_i$.
	Then, we can write $f(x) = g(x) (x-a)^e + \sum_{e_i < e} \alpha_i (x-a)^{e_i} +  \sum_{i = 1}^{p} \beta_i (x - a_i)^{d_i}$. By means of Proposition \ref{repeatedNodesRootSDE}
	we have that $g(x) (x-a)^e$ is a solution of the same $\SDE$. Moreover, for all $e_i > e$, then $(x-a)^{e_i}$ is also a solution of the $\SDE$. But this is not possible
	since the set of solutions is a vector space, and, hence, $(x-a)^e$ would also be a solution to the same $\SDE$.
\end{proof}

\medskip

The proof of the following Corollary is similar to that of Corollary \ref{bigexpCor} but makes use of Corollary \ref{biggapsolution} instead of Proposition \ref{bigexponents}.

\begin{corollary}\label{biggapsCor}
	Let $f \in \FF[x]$ be a polynomial that can be written as 
	\[
		f = \sum_{i = 1}^s \alpha_i (x-a_i)^{e_i}
	\]
	with $a_i, \alpha_i \in \FF$, $e_i > 5 s^2 / 2$ and, whenever $a_i = a_j$ for some $1 \leq i < j \leq s$,
	then $|e_{i} - e_j| > 5s^2 / 2$.
	\begin{enumerate}[\quad a)]
		\item $\{(x-a_i)^{e_i}\, \vert \, 1 \leq i \leq s\}$ are linearly independent,

		\item If $f = \sum_{i = 1}^t \beta_i (x - b_i)^{d_i}$ with $t \leq s$, then $t = s$ and we have the equality
	$\{(\alpha_i,a_i,e_i) \, \vert \, 1 \leq i \leq s\} = \{(\beta_i,b_i,d_i) \, \vert \, 1 \leq i \leq s\}$; in particular,
		$\AffPow_\FF(f) = s$,
	
		\item $f$ satisfies a $\SDE(2s-1,0)$,
		
		\item if $f$ satisfies a $\SDE(k,0)$ with $k \leq 2s-1$, then $(x-a_i)^{e_i}$ also satisfies it for all $i \in \{1,\ldots,s\}$, and
		
		\item $f$ does not satisfy any $\SDE(k,0)$ with $k < s$.
	\end{enumerate}
\end{corollary}

From this corollary we get the following result whose proof is similar to that of Theorem \ref{th-algbigexp}.

\begin{th-algorithm}[Big gaps] \label{th-algbiggaps}
	Let $f \in \FF[x]$ be a polynomial that can be written as 
	\[
		f = \sum_{i = 1}^s \alpha_i (x-a_i)^{e_i}
	\]
	with $a_i, \alpha_i \in \FF$, $e_i > 5 s^2/2$ and whenever $a_i = a_j$ for some $1 \leq i < j \leq s$,
	then $|e_{i} - e_j| > 5 s^2 / 2$.	Then, $\AffPow_\FF(f) = s$.
	Moreover, there is a polynomial time algorithm {\tt Build}$(f)$ that receives $f = \sum_{i = 0}^d f_i x^i \in \FF[x]$ as input and computes the $s$-tuples of nodes 
	$N(f) = (a_1,\ldots, a_s)$,  coefficients  $C(f) = (\alpha_{1},\ldots,\alpha_s)$ and exponents
	$E(f) = (e_1,\ldots,e_s)$. 
	The algorithm {\tt Build}$(f)$ works as follows:

	\begin{enumerate}[\; \bf Step 1.]
		\item Take $r$ the minimum value such that $f$ satisfies a $\SDE(r, 0)$ and compute explicitly one of these $\SDE$.

		\item Compute $B = \{(x-b_i)^{d_i}\, \vert \, 1 \leq i \leq t\}$, the set of all the solutions of the $\SDE$ of the form $(x-b)^d$ 
		with $(r+1)^2/2 \leq e \leq {\rm deg}(f) + (r^2/2)$. 

		\item Determine $\alpha_1,\ldots,\alpha_r$ such that $f = \sum_{i = 1}^r  \alpha_i (x-b_i)^{d_i}$

		\item Output the sets $C(f) = (\alpha_1,\ldots,\alpha_r), \, N(f) = (b_1,\ldots,b_r)$ and $E(f) = (d_1,\ldots,d_r)$. 

	\end{enumerate}
\end{th-algorithm}

\section{The multivariate case} \label{multisec}

This section concerns the study of the multivariate version of the Affine Power model, i.e., we study
expressions of a polynomial $f \in \FF[x_1,\ldots,x_n]$ as
\begin{equation} \label{multimodel2}
f = \sum_{i = 1}^s \alpha_i \ell_i^{\,e_i},
\end{equation}
where $e_i \in \NN$, $\alpha_i \in \FF$ and $\ell_i$ is a (non constant) linear form for all $i$.
We denote by $\AffPow_{\FF}(f)$ the minimum value $s$ such that there exists a representation of the previous form
with $s$ terms. We will study the uniqueness of optimal representations
and propose an algorithm for finding such representations.
{\em In this section only, we work in the black box model:}
we assume that our algorithm has access to $f$ only through a ``black box'' that
outputs $f(x_1,\ldots,x_n)$ when queried 
on an input $(x_1,\ldots,x_n) \in \FF^n$.
This very general model is standard for the study of many problems about 
multivariate polynomials such as, e.g., factorization~\cite{KT90}, sparse interpolation~\cite{benor88,GLL}, 
sparsest shift~\cite{GKL} or Waring decomposition~\cite{Kayal12}.
We also assume that our algorithm has access to $d=\deg(f)$;
the knowledge of an upper bound on $\deg(f)$ would in fact suffice.
As explained in the introduction, our algorithm proceeds by reduction 
to the univariate case: we solve $n$ univariate projections of the multivariate 
problem, and then ``lift'' them to a solution of the multivariate problem.
One (very) minor difficulty is that our univariate algorithms are presented
for polynomials given in dense representation rather than in black box representation. But it is easy to convert from
black box to dense representation:
\begin{remark} \label{blacktodense}
Suppose that we have black-box access to a polynomial $f(x_1,\ldots,x_n)$ of degree $d$. We can obtain the dense representation of the univariate polynomial
$f_1(x_1)=f(x_1,0,0,\ldots,0)$ by querying $f$ 
on $d+1$ distinct inputs of the form $(a_i,0,\ldots,0)$ and interpolating $f_1$
from its values at $a_0,\ldots,a_d$.
\end{remark}
In our algorithm we perform a random change of coordinates before projecting
to a univariate problem. Converting to dense representation in this case
is hardly more difficult:
\begin{remark} \label{blacktodense2}
Suppose that we have black-box access to a polynomial $f(x_1,\ldots,x_n)$ of degree $d$. Let $g(x)=f(\Lambda.x+\lambda)$, 
where $\lambda = (\lambda_1,\ldots,\lambda_n) \in \FF^n$ and $\Lambda=(\lambda_{ij})$ is an $n \times n$ matrix.

We can obtain the dense representation of the univariate polynomial
$g_1(x_1)=g(x_1,0,0,\ldots,0)=f(\lambda_{11}x_1+\lambda_1,\lambda_{21}x_1+\lambda_2,\ldots,\lambda_{n1}x_1+\lambda_n)$ by evaluating $g_1$ at $d+1$ points 
and interpolating from those values.
Equivalently, we can observe that a black-box for $g$ can be constructed from the black box for $f$, and we can therefore apply Remark~\ref{blacktodense} 
to~$g$.
\end{remark}

Having recalled these well-known facts, we proceed with uniqueness considerations. Strictly speaking the optimal expressions in model~(\ref{multimodel2}) are never unique since for all $\lambda \in \FF \setminus \{0\}$ we have
 $\alpha_i \ell_i^{\,e_i} = \beta_i t_i^{\,e_i}$ with $\beta_i := \alpha_i \lambda^{e_i}$ and $t_i := \ell_i/\lambda$. 
To deal with this ambiguity, 
we use the notion of {\em essentially equal} expressions.
Given $f$ we say that two expressions of $f = \sum_{i = 1}^s \alpha_i \ell_i^{\,e_i} =  \sum_{i = 1}^r \beta_i t_i^{\,d_i}$ are
essentially equal if $r = s$ and there exists a permutation $\sigma$ of $\{1,\ldots,s\}$ such that $\alpha_i \ell_i^{\,e_i} =
 \beta_{\sigma(i)} t_{\sigma(i)}^{\,d_{\sigma(i)}}$ for all $i \in \{1,\ldots,s\}$.
Likewise, we say that $f$ 
has an {\em essentially unique} optimal decomposition in
the  multivariate Affine Powers model if two optimal decompositions of $f$
are always essentially equal.
 
If the representation of $f = \sum_{i = 1}^s \alpha_i \ell_i^{\,e_i}$ is optimal, $\ell_i$ and $\ell_j$ cannot be proportional whenever $e_i = e_j$. Otherwise if $\ell_i = \lambda \ell_j$ with $\lambda \in \FF$, we can rewrite $\alpha_i \ell_i^{\,e_i} + \alpha_j \ell_j^{\,e_j} =
 (\lambda^{e_i} \alpha_i + \alpha_j) \ell_j^{\,e_j}$
 
 The following result provides a sufficient condition
 for $f$ to have an essentially unique optimal decomposition in the multivariate Affine Powers model. Indeed, it is an extension to the multivariate
 setting of Corollary \ref{uniquenessExpression}. 
 
\begin{proposition}\label{multiuniqueness}Let $f \in \FF[x_1,\ldots,x_n]$ be a polynomial of the form:
	\[
		f = \sum_{i=1}^s \alpha_i \ell_i^{\, e_i}
	\]
	where $\alpha_i \in \FF \setminus \{0\}$, the $\ell_i$ are non constant linear forms, and $\ell_i$ is not proportional to $\ell_j$ whenever $e_i = e_j$.
	For every $e \in \NN$ we denote by $n_e$ the number of exponents smaller than $e$, i.e., $n_e = \#\{i : e_i \leq e\}$.
	If $n_e \leq \sqrt{\frac{e+1}{2}}$ for all $e \in \NN$, then $\AffPow_\FF(f) = s$ and the optimal representation of
	$f$ is essentially unique.
\end{proposition}
\begin{proof}
Let $r := \AffPow_{\FF}(f) \leq s$ and let $f = \sum_{i = s+1}^{s+r} \alpha_i \ell_i^{\, e_i}$ be an optimal representation of $f$. We write 
$\ell_i = \sum_{j = 1}^n a_{ij} x_j + a_{i0}$ for all $i \in \{1,\ldots,s+r\}$. 
Consider the ring homomorphism 
$\varphi: \FF[x_1,\ldots,x_n] \rightarrow \FF[x]$
induced by $x_i \mapsto \omega_i x + \lambda_i$ where $\omega = (\omega_1,\ldots,\omega_n), \lambda = (\lambda_1,\ldots,\lambda_n) \in \FF^n$.
 If we write $\varphi(\ell_i) = b_i x + c_i$, we choose
$\omega$ and $\lambda$ such that 
\begin{itemize}
\item[(1.a)]  $b_i \neq 0$ and $c_i \neq 0$ for all $i \in \{1,\ldots, r+s\}$, and
\item[(1.b)]  for all $1 \leq i < j \leq s+r$, $\varphi(\ell_i) = \mu \varphi(\ell_j)$ with $\mu \in \FF$
if and only if $\ell_i = \mu \ell_j$.
\end{itemize}
It is important to observe that a generic choice of $\omega, \lambda \in \FF^n$ fulfils these two conditions.
 Then 
$$\begin{array}{lllll} \varphi(f) & = & \sum_{i=1}^s \alpha_i \varphi(\ell_i)^{\, e_i} & = & \sum_{i=1}^s \alpha_i b_i^{\,e_i} (x + c_i/b_i)^{\,e_i} \\ & = & \sum_{i=s+1}^{s+r} \alpha_i \varphi(\ell_i)^{\, e_i} & = & \sum_{i=s+1}^{s+r}  \alpha_i b_i^{\,e_i} (x + c_i/b_i)^{\,e_i}. \end{array} $$
We consider the expression $\varphi(f) = \sum_{i=1}^s \alpha_i b_i^{\,e_i} (x + c_i/b_i)^{\,e_i}$ in the univariate Affine Power model.
By (1.b), whenever $e_i = e_j$ then $c_i/b_i \neq c_j/b_j$. Moreover it satisfies 
that $\{i \in \{1,\ldots,s\}: e_i \leq e\} = n_e \leq \sqrt{\frac{e+1}{2}}$ for all $e \in \NN$. Hence we apply Corollary \ref{uniquenessExpression} to get that
 $r \geq \AffPow_{\FF}(\varphi(f)) = s \geq r$ and that both expressions for $\varphi(f)$ are the same. After reindexing if necessary we get that
\begin{itemize}
\item[(2.a)] $\alpha_i b_i^{\, e_i} = \alpha_{i+s} b_{i+s}^{\, e_{i+s}}$,
\item[(2.b)] $c_i / b_i = c_{i+s} / b_{i+s}$,
\item[(2.c)] and $e_i = e_{i+s}$ for all $i \in \{1,\ldots,s\}$. 
\end{itemize} 
By (2.b) we have that $b_i x + c_i = \mu (b_{i+s}x + c_{i+s})$ with $\mu := b_i / b_{i+s}$. By (1.b) we have that
$\ell_i = \mu \ell_{i+s}$. Finally, by (2.a) and (2.c),  we conclude that $$\alpha_i \ell_i^{\, e_i} = \alpha_i \mu^{\, e_i} \ell_{i+s}^{\,e_i} = \alpha_{i+s} \ell_{i+s}^{\, e_{i+s}},$$
proving that the optimal representation of $f$ is essentially unique.
\end{proof}

Our next goal is to provide algorithms that, given black-box access  to a polynomial $f \in \FF[x_1,\ldots,x_n]$, compute
$s = \AffPow_\FF(f)$ and the terms $\alpha_i \ell_i^{\, e_i}$ for $i \in \{1,\ldots,s\}$ such that $f = \sum_{i = 1}^s \alpha_i \ell_i^{\,e_i}$.
We are going to prove a multivariate analogue of Theorem  \ref{diffnodesimproved} where the condition of "distinct nodes" is replaced by
"the $\ell_i$'s in the decomposition are not proportional". The same strategy that 
we are going to exhibit in the proof also applies to obtain similar results for the other algorithms of
sections 4 and 5.

\begin{th-algorithm} \label{multidiffnodes}
	Let $f \in \FF[x_1,\ldots,x_n]$ be a polynomial that can be written as 
	\[
		f = \sum_{i = 1}^s \alpha_i \ell_i^{\,e_i},
	\]
	where $\ell_i$ are nonconstant linear forms such that $\ell_i \neq \lambda \ell_j$ for all $\lambda \in \FF$, $1 \leq i < j \leq s$, $\alpha_i \in \FF \setminus \{0\}$, and  $e_i \in \NN$. 
Assume that $n_i \leq  (3i/4)^{1/3} - 1$ for all 
$i \geq 2$, where $n_i$ denotes the number of indices $j$ such that
 $e_j \leq i$. Then, $\AffPow_\FF(f) = s$. 

Moreover, there is a randomized algorithm {\tt MultiBuild}$(f)$ that, given access 
to a black box  for $f$ and to $d=\deg(f)$,
computes the set of terms $T(f) = \{\alpha_i \ell_i^{\, e_i} \, \vert \, 1 \leq i \leq s\}$. The algorithm {\tt MultiBuild}$(f)$ runs in time polynomial in $n$ and $d$, and works as follows:

	\begin{enumerate}[\; \bf Step 1.]
		\item We define $g := \phi(f)$ where $\phi$ is a random affine change of coordinates ($x_i \mapsto \sum_{j = 1}^n \lambda_{ij} x_j + \lambda_i$ for all $i$).
				
		\item For each $j \in \{1,\ldots,n\}$, we set $g_j := \pi_j(g)$ where $\pi_j: \FF[x_1,\ldots,x_n] \longrightarrow \FF[x]$ 
		is induced by $x_k \mapsto 0$ if $k \neq j$ and $x_j \mapsto x$. 
		
		We apply the algorithm {\tt Build}$(g_j)$ from Theorem \ref{diffnodesimproved}
to obtain $s_j := \AffPow_{\FF}(g_j)$
		and the triplets $(\beta_{ij}, b_{ij}, e_{ij}) \in \FF \times \FF \times \NN$ such that 
		$g_j = \sum_{i = 1}^{s_j} \beta_{ij} (x+b_{ij})^{e_{ij}}$.
		 
		If there exist $i,j$ such that $b_{ij} = 0$, then output 'It is not possible to reconstruct $f$'.
		Otherwise, for all $j$ we define the set of triplets
		$$P_j := \{(c_{ij}, p_{ij}, e_{i,j}) \, \vert \, c_{ij} := \beta_{ij} b_{ij}^{e_{ij}},\ p_{ij} := 1/b_{ij},\ 1 \leq i \leq s_i\}.$$

		\item If one of these conditions holds:
		\begin{enumerate}
			\item there exist $j_1 \neq j_2$ such that $s_{j_1} \neq s_{j_2}$, 
			\item there exist $i_1 \neq i_2$ and $j$ such that $c_{i_1j} = c_{i_2j}$, or
			\item there exist $i, j$ such that for all $i'$, $c_{i1} \neq c_{i'j}$ or $e_{i1} \neq e_{i'j}$;
		\end{enumerate}
		then output: 'It is not possible to reconstruct $f$'. Otherwise we set $s := s_1 = s_2 = \cdots = s_r$ and reorder the elements of $P_2,\ldots,P_n$ so that
		$c_i := c_{i1} = c_{i2} = \cdots = c_{in}$ and $e_i := e_{i1} = e_{i2} = \cdots = e_{in}$ for all $i \in \{1,\ldots,s\}$.		
		
		\item $g = \sum_{i = 1}^s c_i (1 + \sum_{j = 1}^n p_{ij} x_j)^{e_i}$, so we output $f = \sum_{i = 1}^s c_i (\phi^{-1}(1 + \sum_{j = 1}^n p_{ij} x_j))^{e_i}$
	\end{enumerate}
	
	If the $\lambda_i$'s and the $\lambda_{ij}$'s needed to define $\phi$ are chosen uniformly at random from a finite set $S$, then the probability of success of the algorithm is at
	least $$1 - \frac{d^{2/3}(2n+d)}{|S|}.$$
\end{th-algorithm}

\begin{proof}
	The input polynomial $f$ satisfies the hypotheses of Proposition \ref{multiuniqueness}, so $\AffPow_{\FF}(f) = s$ and the optimal representation of $f$
	is essentially unique. 
	
	After applying a random $\phi$ as described in {\bf Step 1}, with high probability\footnote{A detailed probabilistic analysis is performed at the end of this proof.} we have that $\phi$ is invertible and 
  	$g = \sum_{i = 1}^s \alpha_i t_i^{\, e_i}$ with $t_i = \sum_{j = 1}^n a_{ij} x_j + a_{i0}$ satisfies the following properties:
  	\begin{itemize}
  		\item[(i)] $a_{ij} \neq 0$ for all $i,j$.
  		\item[(ii)] for all $j \neq 0$, then $a_{ij}/a_{i0} \neq a_{i'j}/a_{i'0}$ for all $i,i'$, and
  		\item[(iii)] $\alpha_i a_{i0}^{e_i} \neq \alpha_{i'}a_{i'0}^{e_{i'}}$ for all $i \neq i'$.
  	\end{itemize} 
  	It is important to observe that for a generic choice of the $\lambda_i$'s and $\lambda_{ij}$'s involved in the definition of $\phi$, these conditions will be fulfilled.
  	The goal of the algorithm is to recover $f$ via the following expression of $g$:    		
  		\[ g = \sum_{i = 1}^s \alpha_i a_{i0}^{e_i} \left(1 + \sum_{j = 1}^n \frac{a_{ij}}{a_{i0}} x_j\right)^{e_i}; \]
  	so we are interested in computing the values 
  	\begin{itemize}
  		\item $\alpha_i a_{i0}^{e_i}$ for all $i$
  		\item $a_{ij}/a_{i0}$ for all $i,j$
  		\item $e_i$ for all $i$
  	\end{itemize}
 	In {\bf Step 2}, for all $j \in \{1,\ldots,n\}$ we consider 
 	\[ 
 		\pi_j(g) = \sum_{i = 1}^s \alpha_i a_{i0}^{e_i} \left(1 + \frac{a_{ij}}{a_{i0}}\, x\right)^{e_i} =
 		\sum_{i = 1}^s \alpha_i a_{ij}^{e_i} \left(x+\frac{a_{i0}}{a_{ij}}\right)^{e_i}.
 	\]
 	Since $\pi_j(g)$ satisfies the hypotheses of  Theorem \ref{diffnodesimproved}
{\tt Build}$(\pi_j(g))$ outputs the values
 		$$\left\{ (\alpha_i a_{ij}^{e_i},\frac{a_{i0}}{a_{ij}},e_i) \, \vert \, 1 \leq i \leq s \right\}.$$
 	From these values we obtain in the sets
	 	$$P_j = \left\{ (\alpha_i a_{i0}^{e_i},\frac{a_{ij}}{a_{i0}},e_i) \, \vert \, 1 \leq i \leq s \right\}.$$
Before calling {\tt Build}$(\pi_j(g))$, we compute the dense representation
of $\pi_j(g)$ using Remarks~\ref{blacktodense} and~\ref{blacktodense2}.

	 Thanks to the unique expression of $g_j$ for all $j$ and to (iii)
	 we have that none of the conditions of {\bf Step 3} is satisfied and we obtain $g$ in {\bf Step 4}.
	
	If we see the values of $\lambda_i, \lambda_{ij}$ used to define $\phi$ as variables, 
the invertibility of $\phi$ is equivalent
	to the nonvanishing of a degree $n$ polynomial. Moreover, the $a_{ij}$ are degree one polynomials in these variables. Thus,
	the conditions $a_{ij} \neq 0$ consist in the nonvanishing of $s(n+1)$ polynomials of degree $1$.
	The conditions $a_{ij}/a_{i0} \neq a_{i'j}/a_{i'0}$ for all $i,i',j$ with $j \neq 0$ can be seen as the
	nonvanishing of $s(s-1)n/2$ polynomials of degree $2$. The conditions  $\alpha_i a_{i0}^{e_i} \neq \alpha_{i'}a_{i'0}^{e_{i'}}$ can be seen as the nonvanishing of
	$s(s-1)/2$ polynomials of degree at most ${\rm max}(e_i)$, which, by Corollary \ref{upperbounddegreeterms}, is upper bounded by $d + (s^2/2)$. Hence, all the
	conditions to be satisfied can be codified in a nonzero polynomial $\psi$ of degree 
		$$n + s(n+1) + s(s-1)n + (s(s-1)(2d+s^2)/4) \leq  \frac{8s^2 n + 2 s^2 d + s^4}{4}.$$
	Moreover, if we set $e := {\rm max}(e_i)$, then
	\begin{itemize}
		\item $e \leq d + (s^2/2)$, and
		\item $s = n_e \leq (3e/4)^{1/3},$
	\end{itemize}
	form where we deduce that $s \leq d^{1/3}$ and the degree of $\psi$ is upper bounded by $ d^{2/3} (2n+d)$.
	Hence, by the Schwartz-Zippel lemma, if we assume the $\lambda_i, \lambda_{ij}$ are taken uniformly at random from a finite set $S$, the probability of satisfying
	all these constraints is at least
		$$1 - \frac{d^{2/3}(2n+d)}{|S|}.$$
	and the result follows.
\end{proof}

{\small

\section*{Acknowledgments}

The reconstruction problem for sums of affine powers was suggested to one of us (P.K.) by Erich Kaltofen at a Dagstuhl workshop where P.K. gave a talk on lower bounds for this model.

}
\appendix

\section{Appendix: Algorithms for Sparsest Shift and Waring decomposition}

In this appendix we apply the techniques from 
the previous sections to study optimal decompositions of polynomials 
in the Waring and Sparsest Shift models. 
As explained in the introduction, these two models have been 
extensively studied in the literature.
We do not claim that the algorithms proposed in this appendix improve on the
existing methods. Rather, we present them for the sole purpose of illustrating
on these two classical models the techniques developed for the more
general model of sums of affine powers.

\medskip

\subsection{Waring decompositions} \label{waring-appendix}

In Proposition \ref{smallequation}
 we saw that if $f$ has an expression in the $\AffPow$ model with $s$ terms, then $f$ satisfies a
 $\SDE(2s-1,0)$. We begin this section by proving that 
an expression of $f$ with $s$ terms in the Waring model yields a $\SDE$ satisfied by $f$ of order $s$ and shift $0$.

\begin{proposition}\label{shift0Waring}
	Let $f \in \FF[x]$ be written as
	\[
		f = \sum_{i = 1}^s \alpha_i (x - a_i)^{d},
	\]
	Then $f$ satisfies a $\SDE(s, 0)$ that is also satisfied by the $(x - a_i)^{d}$'s.
\end{proposition}
\begin{proof}
	We consider the $\SDE$ in the unknown $g$ given by the Wronskian: 
	\begin{equation}
		\Wr(g, (x-a_1)^{d}, \ldots, (x - a_s)^{d})(x) = 0
	\end{equation}
	After factoring out $(x-a_i)^{d-s}$ for all $i$, we get the reduced $\SDE$:
	\[
		\sum_{i=0}^s R_i(x) g^{(i)}(x) = 0,
	\]
	where
	\[
		R_i =
		\begin{vmatrix}
			(x-a_1)^{s} 						& \dots 	& (x-a_s)^{s}\cr 
			d^{\underline{1}}(x - a_1)^{s - 1}  & \dots 	& d^{\underline{1}}(x - a_s)^{s - 1} \cr 
			\vdots 								& \ddots	& \vdots		 \cr
			d^{\underline{i-1}}(x - a_1)^{s-i+1}& \dots 	& d^{\underline{i-1}}(x - a_s)^{s-i+1} \cr 
												&			&	\cr			
			d^{\underline{i+1}}(x - a_1)^{s-i-1}& \dots 	& d^{\underline{i+1}}(x - a_s)^{s-i-1} \cr 
			\vdots 								& \ddots	& \vdots		 \cr
			d^{\underline{s}}		 			& \dots 	& d^{\underline{s}} 
		\end{vmatrix}
	\]
	and $d^{\underline{k}} := \prod_{j = 1}^{k} (d - j + 1)$. 	Because of the nice structure induced by all the exponents being equal to $d$, we have that $R_{i+1}' = R_i$, and hence ${\rm deg}(R_i) = {\rm deg}(R_s) - (s-i)$.
	If we factor out the constants on each row in $R_s$ we get that
	\[
		R_s = \prod_{i=1}^s d^{\underline{i}} \cdot 
		\begin{vmatrix}
			(x-a_1)^{s} 		& \dots 	& (x-a_s)^{s}		\cr 
			(x - a_1)^{s - 1}	& \dots 	& (x - a_s)^{s - 1} \cr 
			\vdots				& \ddots	& \vdots 			\cr
			(x - a_1)  			& \dots 	& (x - a_s)
		\end{vmatrix}
	\]
	We factor $(x-a_i)$ on each row and use the known formula for the determinant of a Vandermonde matrix to obtain:
	\[
		R_s = \prod_{i=1}^s d^{\underline{i}} \cdot \prod_{i=1}^s (x - a_i) \cdot \prod_{i < j} (a_i - a_j) 
	\]
	We have $\deg(R_s) = s$, and hence $\deg(R_i) = i$, which shows that the reduced $\SDE$ has in fact a zero shift.
\end{proof}

Moreover, when $\Waring_\FF(f)$ is small enough we have that the $\SDE$ provided in Proposition \ref{shift0Waring} is the only
$\SDE(s,0)$ satisfied by $f$.

\begin{corollary}
	Let $f$ be a polynomial such that $s = \Waring_\FF(f) \leq \sqrt{2d/3}$.
	Then $f$ satisfies a unique $\SDE(s, 0)$.
\end{corollary}
\begin{proof}
	Consider a $\SDE(s, 0)$ satisfied by $f$:
	\[
		\sum_{i = 0}^s P_i(x) f^{(i)}(x) = 0
	\]
	Since $\Waring_\FF(f) = s$, $f$ can be expressed as $f = \sum_{i=1}^s \alpha_i(x - a_i)^d$ and $\{(x-a_i)^d\, :\, 1 \leq i \leq d\}$ are linearly independent.
	By Proposition \ref{bigexponents} we get that any term $(x - a_i)^d$ satisfies this $\SDE$ because
	\[
		d \geq \frac{3}{2} s^2 \geq s + s(s-1) + \binom{s}{2}.
	\]
	Thus, we apply Lemma \ref{uniqueSDE} to conclude the equation given by the Proposition \ref{shift0Waring} is the unique $\SDE(s,0)$ satisfied by $f$.
\end{proof}

\medskip

As a direct consequence of this result, we have the following algorithm.

\begin{algorithm}
	Let $f$ be a polynomial of degree $d$. There is a polynomial time algorithm {\tt WaringDec}$(f)$ that receives $f = \sum_{i = 0}^d f_i x^i \in \FF[x]$ as input
	decides if $\Waring(f) \leq \sqrt{2d/3}$. Moreover, whenever $\Waring_\FF(f) \leq \sqrt{2d/3}$, with the optimal decomposition being
	\[
		f = \sum_{i = 1}^s \alpha_i (x - a_i)^{d},
	\]
	the algorithm computes the $s$-tuples of shifts $S(f) = (a_1,\ldots,a_s)$ and coefficients $C(f) = (\alpha_1,\ldots,\alpha_s)$.
	The algorithm works as follows:

	\begin{enumerate}[\; \bf Step 1.]
		\item Find the minimum $k$ such that there exists an $\SDE(k, 0)$ satisfied by $f$ and compute explicitly one of these SDE.

		\item If $k > \sqrt{2d/3}$, then $\Waring_{\FF}(f) > \sqrt{2d/3}$.

		\item Compute the set $\mathcal B = \{(x-b_i)^d \, \vert \, 1 \leq i \leq t\}$ of solutions of the form $(x-a)^d$ of this $\SDE$.
		
		\item If $t < k$, then $\Waring_{\FF}(f) > \sqrt{2d/3}$.
		
		\item Write $f = \sum_{i = 1}^t \beta_i (x-b_i)^d,$ and output $S(f) = (b_1,\ldots,b_t)$ and $C(f) = (\beta_1,\ldots,\beta_t)$.
	\end{enumerate}
\end{algorithm}
Note that if we reach Step 5 of the algorithm, we have $t=k=s \leq \sqrt{2d/3}$.

\subsection{Sparsest Shift decompositions}

We saw in Section~\ref{waring-appendix} that a polynomial with a Waring decomposition of size $s$ satisfies a $\SDE$ of order $s$ and shift 0. 
The same is true for the Sparsest Shift model:

\begin{proposition}
	Let $f \in \FF[x]$ be written as
	\[
		f = \sum_{i = 1}^s \alpha_i (x - a)^{e_i},
	\]
	Then $f$ satisfies a $\SDE(s, 0)$ that is also satisfied by the $(x - a)^{e_i}$'s.
\end{proposition}
\begin{proof}
	We will prove something stronger, namely that $f$ satisfies an $\SDE(s,0)$ of the following form
	\[
		\sum_{i=0}^s \gamma_i (x-a)^i g^{(i)}(x) = 0,
	\]
	where $\gamma_0,\ldots,\gamma_s \in \FF$.
	We take the original $\SDE$ given by the Wronskian of an unknown polynomial $g$ and $(x-a)^{e_i}$ for all $i \in \{1,\ldots,s\}$:
	\[
		\sum_{i=0}^s (-1)^i\, P_i(x) \, g^{(i)}(x) = 0.
	\]
	Because of the stepped sequence of degrees in the determinant defining $P_i$, there exists an integer $\Delta_i$ such that every permutation $\sigma$ corresponds to a term $c_
	\sigma (x - a)^{\Delta_i}$.
	More precisely, we have
	\[
		\Delta_i = \left( \sum_{j=1}^s e_j \right) - \binom{s}{2} + i
	\]
	Thus, the determinant is either 0, or some constant times $(x - a)^{\Delta_i}$.
	Moreover, we have $\Delta_{i + 1} = \Delta_i + 1$ and hence we can rewrite the $\SDE$ as
	\[
		\sum_{i=0}^s c_i (x-a)^{\Delta_0 + i} g^{(i)}(x) = 0 
	\]
	with $c_i \in \FF$.
	We factorize this equation by $(x - a)^{\Delta_0}$ to obtain the wanted $\SDE(s, 0)$.
	Notice that the factorization again doesn't change the space of solutions, hence $f$ and $(x-a)^{e_i}$ are still solutions of this $\SDE$.
\end{proof}

\begin{corollary}
	Let $f$ be a polynomial of degree $d$ such that $s = \Sparsest_\FF(f) < \sqrt{d}$, 
and let $a \in \FF$ be the corresponding sparsest shift.
Let $\sum_{i = 1}^k P_i(x)f^{(i)}(x) = 0$ be any $\SDE(k, 0)$ satisfied by $f$.
 If $k \leq s$ we must have $P_k(a) = 0$.
\end{corollary}
\begin{proof}
	Assume that $f$ satisfies an $\SDE(k, 0)$ : $\sum_{i = 1}^k P_i(x)f^{(i)}(x) = 0$ with $k \leq s$.
	The sparsest shift decomposition of $f$ is
	\[
		f = \sum_{i = 1}^s \alpha_i (x - a)^{e_i}.
	\]
	Assume that $e_1 > e_2 > \dots > e_s > e_{s+1} := -1$.
	We take $t \in \{1,\dots, s\}$ the maximum value such that $e_t - e_{t + 1} - 1 \geq s$.
	Such a value $t$ exists, otherwise $d \leq e_1 < \sum_{i = 1}^s (e_i - e_{i+1}) < s^2$, a contradiction.
	We rewrite $f$ as
	\[
		f = (x - a)^{e_t}g(x) + \sum_{i = t + 1}^s \alpha_i (x - a)^{e_i}.
	\]
	Now we apply Proposition \ref{repeatedNodesRootSDE} with $p = 0$ to get that $(x - a)^{e_t}g$ satisfies the same $\SDE$ because $e_t - e_{t + 1} > s \geq k$.
	By the same argument as in Proposition \ref{bigExpRepeatedNodesRootSDE}, we conclude that $P_k(a) = 0$.
\end{proof}

\begin{algorithm}
	Let $f$ be a polynomial of degree $d$. There is a polynomial time algorithm {\tt SparsestShift}$(f)$ that receives $f = \sum_{i = 0}^d f_i x^i \in \FF[x]$ as input
	decides if $\Sparsest_\FF(f) \leq \sqrt{d}$; moreover, whenever $\Sparsest_\FF(f) \leq \sqrt{d}$, with the optimal decomposition being
	\[
		f = \sum_{i = 1}^s \alpha_i (x - a)^{e_i},
	\]
	the algorithm computes the shift $a \in \FF$, and the $s$-tuples of coefficients $C(f) = (\alpha_1,\ldots,\alpha_s)$ and exponents $E(f) = (e_1,\ldots,e_s)$.
	The algorithm works as follows:

	\begin{enumerate}[\; \bf Step 1.]
		\item Find the minimum $k$ such that there exists an $\SDE(k, 0)$ satisfied by $f$ and compute explicitly one of these SDE. Namely, 
			\[
				\sum_{i = 0}^k P_i(x) f^{(i)}(x) = 0
			\]		
		\item Factorize the last coefficient of this $\SDE$, i.e., write: $$P_k = c \cdot \prod_{i = 1}^k (x - a_i).$$
		
		\item For each $a_i$, 
decompose $f$ in the shifted basis $((x-a_i)^j)_{0 \leq j \leq d}$.
		
		\item If the decomposition with smallest number of terms has $\leq \sqrt{d}$ terms, we output this decomposition; otherwise, $\Sparsest_\FF(f) > \sqrt{d}$.
	\end{enumerate}
\end{algorithm}


\begin{thebibliography}{00}

\bibitem{alexander95}
James Alexander and Andr{\'e} Hirschowitz.
\newblock Polynomial interpolation in several variables.
\newblock {\em Journal of Algebraic Geometry}, 4(2):201--222, 1995.

\bibitem{benor88}
M. Ben-Or and P. Tiwari. 
\newblock A deterministic algorithm for sparse multivariate polynomial interpolation.
\newblock 
In {\em Proc. 20th annual ACM Symposium on Theory of Computing}. ACM, 1988.



\bibitem{Bocher} 
M. Bocher. 
The theory of linear dependence. {\em Annals of Mathematics},
2(1/4): 81--96, 1900-1901.

\bibitem{BoTi}
A. Borodin and P. Tiwari. On the decidability of sparse univariate polynomial interpolation. {\em Computational Complexity}, 1(1):67-90, 1991.

\bibitem{BCG} 
M. Boij, E. Carlini, A.~V. Geramita.
Monomials as sums of powers: the real binary case.  
{\em Proc. Amer. Math. Soc.} 139(9):3039--3043, 2011. 


\bibitem{Bostanetal} A. Bostan, F. Chyzak, M. Giusti, R. Lebreton, G. Lecerf, B. Salvy, and \'E. Schost.
Algorithmes efficaces en calcul formel. Preprint version, 2017, 686 pages.


\bibitem{brambilla08}
Maria~Chiara Brambilla and Giorgio Ottaviani.
\newblock On the {Alexander} --{Hirschowitz} theorem.
\newblock {\em Journal of Pure and Applied Algebra}, 212(5):1229--1251, 2008.

\bibitem{Cox}
D. A. Cox, 
Galois theory. Second edition. 
Pure and Applied Mathematics (Hoboken). John Wiley \& Sons, Inc., 2012.


\bibitem{GK} 
I. Garc\'{\i}a-Marco, and P. Koiran.
Lower bounds by Birkhoff interpolation. 
Submitted, arXiv:1507.02015 [cs.CC]. 


\bibitem{GKL} M. Giesbrecht, E. Kaltofen, W. Lee.
\newblock Algorithms for computing sparsest shifts of polynomials in power, Chebyshev and Pochhammer bases. 
\newblock {\em International Symposium on Symbolic and Algebraic Computation} (ISSAC'2002) (Lille). {\em Journal of  Symbolic Computation} 36(3-4):401--424, 2003.

\bibitem{GLL}
M. Giesbrecht, G. Labahn and W.-S. Lee. Symbolic-numeric sparse interpolation of multivariate polynomials. 
{\em Journal of Symbolic Computation} 44(8):943--959,  2009.

\bibitem{GR} M. Giesbrecht, D.~S. Roche.
Interpolation of shifted-lacunary polynomials. 
{\em Computational  Complexity} 19(3):333--354, 2010.

\bibitem{GK93} D. Grigoriev and M. Karpinski. A zero-test 
and an interpolation algorithm for the shifted sparse
               polynomials.
 In {\em Proc. Applied Algebra, Algebraic Algorithms and Error-Correcting Codes,
               10th International Symposium (AAECC-10)}. 
LNCS 673, pp. 162-169, Springer, 1993.


\bibitem{GrY} J.~H. Grace, A. Young. The algebra of invariants. Cambridge
University Press, 1903.

\bibitem{IaKa} A.~Iarrobino, V.~Kanev, Power sums, Gorenstein algebras, and determinantal loci. Appendix C by Iarrobino and Steven L. Kleiman. Lecture Notes in Mathematics, 1721. Springer-Verlag, Berlin, 1999. 


\bibitem{KT90} E. Kaltofen and B.  Trager. 
\newblock Computing with polynomials given by black boxes for their evaluations: Greatest common divisors, factorization, separation of numerators and denominators.
\newblock {\em  Journal of Symbolic Computation} 9(3):301-320, 1990.

\bibitem{Kleppe}
J. Kleppe.
\newblock {Representing a Homogenous Polynomial as a Sum of Powers of Linear
  Forms}.
\newblock {\em Thesis for the degree of Candidatus Scientiarum (University of
  Oslo)}, 1999.
\newblock Available at \url{http://folk.uio.no/johannkl/kleppe-master.pdf}.

\bibitem{Kayal12} N.~Kayal. Affine projections of polynomials. In {\em Proc. 44th annual ACM Symposium on Theory of Computing (STOC 2012)}, pp. 643-662. 
ACM, 2012.

\bibitem{KKPS}
N.~Kayal, P.~Koiran, T.~Pecatte, and C.~Saha.
Lower bounds for sums of powers of low degree univariates.
In {\em Proc. 42nd International Colloquium on Automata, Languages
  and Programming (ICALP 2015), part~I}, LNCS 9134, pages 810--821. Springer,
  2015.
Available from \url{http://perso.ens-lyon.fr/pascal.koiran}.

\bibitem{LS}
Y.~N. Lakshman, B.~D. Saunders, Sparse shifts for univariate polynomials.
{\em Appl. Algebra Engrg. Comm. Comput.} 7 (1996), no. 5, 351--364.

\bibitem{landsberg2010}
Joseph~M Landsberg and Zach Teitler.
\newblock On the ranks and border ranks of symmetric tensors.
\newblock {\em Foundations of Computational Mathematics}, 10(3):339--366, 2010.


\bibitem{LLL} 
A.~K. Lenstra, H.~W. Lenstra, L. Lov\'asz.
Factoring polynomials with rational coefficients. 
Math. Ann. 261 (1982), no. 4, 515--534.

\bibitem{PS76}
G. P\'olya, and G. Szeg\"o.
Problems and theorems in analysis. Vol. II. Theory of functions, zeros, polynomials, determinants, number theory, geometry. 
Revised and enlarged translation by C. E. Billigheimer of the fourth German edition. 
Springer Study Edition. Springer-Verlag, New York-Heidelberg, 1976. xi+391 pp.

\bibitem{Schrijver}
A. Schrijver
Theory of linear and integer programming. 
John Wiley \& Sons,  1986. xii+471 pp.


\bibitem{WronskianLemma}
M. Voorhoeve, and A.J. Van Der Poorten.
Wronskian determinants and the zeros of certain functions.
Indagationes Mathematicae, 37 (1975), no. 5, 417--424.
\end{thebibliography}
\end{document}